\documentclass[11pt,a4paper]{article}

\usepackage[margin=1in]{geometry}
\usepackage[main=english, russian]{babel}
\usepackage{amsmath,amssymb,amsthm,mathtools}
\usepackage{mathrsfs}
\usepackage{enumitem}
\usepackage{hyperref}
\usepackage{microtype}
\usepackage{tikz-cd}
\usetikzlibrary{quotes} 
\usetikzlibrary{cd}
\usepackage{cleveref}
\usepackage{epigraph}

\usepackage[OT2, T1]{fontenc}

\hypersetup{
colorlinks=true,
linkcolor=blue,
citecolor=blue,
urlcolor=blue
}

\newcommand{\Diff}{\operatorname{Diff}}
\newcommand{\Cartan}{\operatorname{Cartan}}
\newcommand{\Dens}{\operatorname{Dens}}
\newcommand{\End}{\operatorname{End}}
\newcommand{\Hom}{\operatorname{Hom}}

\newcommand{\Bc}{{\mbox{\selectlanguage{russian}B}}}
\newcommand{\bc}{{\mbox{\selectlanguage{russian}b}}}

\newcommand{\A}{\mathcal A}
\newcommand{\M}{\mathcal M}
\newcommand{\Diffscr}{\mathscr{D}iff}
\newcommand{\Cartanscr}{\mathscr{C}artan}

\newtheorem{theorem}{Theorem}[section]
\newtheorem{lemma}[theorem]{Lemma}
\newtheorem{proposition}[theorem]{Proposition}
\newtheorem{corollary}[theorem]{Corollary}
\theoremstyle{definition}
\newtheorem{definition}[theorem]{Definition}
\newtheorem{remark}[theorem]{Remark}
\newtheorem{notation}[theorem]{Notation}
\newtheorem{example}[theorem]{Example}

\title{\textbf{
Batalin--Vilkovisky algebras from \v{S}evera bicomplexes
}}

\author{
  Eugenia Boffo\thanks{Mathematical Institute, Faculty of Mathematics and Physics, Charles University, Sokolovsk\'a 49/83 -- 18675 Prague -- Czechia -- {\tt boffo@karlin.mff.cuni.cz}} \and
  Domenico Fiorenza\thanks{Dipartimento di Matematica ``Guido Castelnuovo'' --
Sapienza Universit\`a di Roma --
P.le Aldo Moro, 2 --
00185 --
Roma --
Italy -- {\tt domenico.fiorenza@uniroma1.it}
}}

\begin{document}

\maketitle

\epigraph{{\small \emph{Nella parte del mondo in cui sono nato, in cui sono nato\\
Tutto \`{e} gi\`{a} stato detto\\
Tutto \`{e} gi\`{a} stato pensato}}}{{\small I Cani}}

\begin{abstract}
A \v{S}evera bicomplex is a bicomplex $(M,d_0,d_1)$ whose differentials $d_i$ are Grothendieck differential operators of order $\leq i$ with respect to a given $A$-module structure on $M$, where $A$ is a graded commutative algebra. One also requires that $(M,d_0)$ admits a contracting homotopy $K$ that is a differential operator of the same order as $d_0$. Under suitable assumptions, the differential $d_2$ of the second page of the spectral sequence associated with $(M,d_0,d_1)$ is a Batalin--Vilkovisky operator on $E_2$, so that when $E_2$ is a free rank 1 $A$-module, the graded commutative algebra structure of $A$ is enhanced to a Gerstenhaber algebra structure, and the choice of a $d_2$-closed basis element for $E_2$ further enhances this to a Batalin--Vilkovisky algebra structure. The prototypical example of this construction is \v{S}evera’s description of the Batalin--Vilkovisky algebra structure on the algebra of smooth functions of an odd symplectic manifold. Throughout the whole article we look at Batalin-Vilkovisky algebras through the lenses of Cartan calculus. Indications of a generalization to derived Cartan calculus are briefly discussed in the concluding section.

\end{abstract}

\tableofcontents

\section{Grothendieck differential operators}
\label{sec:1}

Let $k$ be a characteristic zero field. let $A$ be a graded commutative unital $k$-algebra and let $M,N$ be graded $A$-modules.
We begin with the usual definition of differential operators in the
sense of Grothendieck \cite{Grothendieck_EGA4_1}, see also \cite{Berthelot1978NotesOC}.

For $a\in A$, denote by
\[
\rho_a^M\colon M\longrightarrow M,
\qquad
\rho_a^M(m)=am
\]
the multiplication operator by $a$. When there is no danger of
confusion, we simply write $\rho_a$.

\begin{definition}[Grothendieck, EGA IV]\label{def:diff}
Set
\[
\Diff_A^{\leq -1}(M,N):=0.
\]
For $n\geq 0$, define inductively
\[
\Diff_A^{\leq n}(M,N)
:=
\left\{
\Bc\in\Hom_{\mathbb Z}(M,N)
\mid
[\Bc,\rho_a]\in
\Diff_A^{\leq n-1}(M,N)
\text{ for all }a\in A
\right\},
\]
where
\[
[\Bc,\rho_a]
=\Bc\circ\rho_a^M-(-1)^{\deg(\Bc)\deg(\rho_a)}\rho_a^N\circ \Bc=
\Bc\circ\rho_a^M-(-1)^{\deg(\Bc)\deg_A(a)}\rho_a^N\circ \Bc,
\]
and set 
\[
\Diff_A(M,N)=\bigcup_n \Diff_A^{\leq n}(M,N).
\]
\end{definition}
\begin{remark}
By construction,  $\Diff_A(M,N)$ is a filtered graded $A$-module, where $A$ acts by left multiplication. 
\end{remark}
\begin{remark}\label{rem:order-zero}
For $n=0$, the definition gives
\begin{align*}
\Diff_A^{\leq0}(M,N)
& =
\left\{
\Bc\in\Hom_{\mathbb Z}(M,N)
\mid
\Bc(am)=(-1)^{\deg(\Bc)\deg_A(a)} a\Bc(m)
\text{ for all }a\in A \text{ and }m\in M 
\right\} \\
& 
=
\Hom_A(M,N).
\end{align*}
Thus degree zero
operators of order at most zero are precisely $A$-module morphisms from $M$ to $N$ (of arbitrary degree).
\end{remark}
\begin{remark}\label{rem:derived-bracket}
One easily sees inductively that an operator $\Bc$ has order at most $n$ precisely when its $(n+1)$-fold
iterated commutators with multiplication operators vanish, i.e.,  for $\Bc\in\Hom_{\mathbb Z}(M,N)$, the following are equivalent:
\[
\Bc\in\Diff_A^{\leq n}(M,N)
\]
and
\[
[\cdots[[\Bc,\rho_{a_0}],\rho_{a_1}],\ldots,\rho_{a_n}]=0
\]
for every $a_0,\ldots,a_n\in A$. Equivalently, recalling Remark \ref{rem:order-zero}, operators of order at most $n$ are precisely those whose $n$-fold
iterated commutators with multiplication operators is an $A$-module morphism from $M$ to $N$:
\[
[\cdots[[\Bc,\rho_{a_1}],\rho_{a_2}],\ldots,\rho_{a_n}]\in \Hom_{A} (M,N).
\]
 We then have a natural map
\[
\Diff_A^{\leq n}(M,N)\otimes A^{\otimes n}\to \Hom_A(M,N),
\]
and so every element $\Bc$ in $\Diff_A^{\leq n}(M,N)$ defines an operator 
\[
\{-\}_{\Bc}\colon A^{\otimes n}\to \Hom_A(M,N)[\deg\Bc]
\]
given by
\[
\{a_1,\dots,a_n\}_{\Bc}=[\cdots[[\rho_{a_1},\Bc],\rho_{a_2}],\ldots,\rho_{a_n}]=-(-1)^{\deg_A(a_1)\deg(\Bc)}[\cdots[[\Bc,\rho_{a_1}],\rho_{a_2}],\ldots,\rho_{a_n}].
\]
This can be thought of as a derived bracket with values in $\Hom_A(M,N)$. 
We will come back to this in what follows. The choice of the prefactor $-(-1)^{\deg_A(a)\deg(\Bc)}$, i.e., of having the commutator $[\Bc, \rho_{a_1}]$ instead of the more natural $[\Bc,\iota]$, is dictated by the request of having better behaved signs later, see Example \ref{ex:order2} and the subsequent Lemmas.

\end{remark}
\begin{remark}\label{rem:sequence}
Up to the prefactor $-(-1)^{\deg_A(a_1)\deg(\Bc)}$, the bracket $\{-\}_{\Bc}$ from Remark \ref{rem:derived-bracket} is actually the $n$-th in a sequence of morphisms of $k$-vector spaces
\[
{\lambda^{(i)}_\Bc} \colon A^{\otimes i}\to \Diff_A^{\leq n-i}(M,N), \qquad \text{for }i=0,\dots,n,
\]
where ${\lambda^{(0)}_\Bc}=\Bc$ and 
\[
{\lambda^{(i+1)}_\Bc}(a_1,\dots,a_i,a_{i+1})=[{\lambda^{(i)}_\Bc}(a_1,\dots,a_i),\rho_{a_{i+1}}].
\]
Since $\Diff_A^{\leq n-i}(M,N)\subseteq \Hom_k(M,N)$, we can think of all of the $\lambda^{(i)}_\Bc$ as taking values in $\Hom_k(M,N)$. Doing this one easily sees that an element in $\Diff_A^{\leq n}(M,N)$ of degree $k$ is equivalent to a sequence of 
morphisms of $k$-vector spaces of degree $k$
\[
\lambda^{(i)}\colon A^{\otimes i}\to \Hom_k(M,N[k]), \qquad \text{for }i=0,\dots,n
\]
such that
\begin{itemize}
    \item $\lambda^{(i+1)}(a_1,\dots,a_i,a_{i+1})=[\lambda^{(i)}(a_1,\dots,a_i),\rho_{a_{i+1}}]$ for any $i=0,\dots, n-1$;
    \item $\lambda^{(n)}(a_1,\dots,a_n)\in \Hom_{A}(M,N[k])$, for any $a_1,\dots, a_n$ in $A$.
\end{itemize}
In agreement with the notation in Remark \ref{rem:derived-bracket}, we write
\[
\mathcal{L}^{(i)}(a_1,\dots,a_i)=-(-1)^{\deg_A(a_1)k}\lambda^{(i)}(a_1,\dots,a_i)
\]
for $i\geq 1$, will write simply $\mathcal{L}$ for $\mathcal{L}^{(1)}$, and will write such a sequence as
\[
(\Bc,\mathcal{L},\mathcal{L}^{(2)},\dots, \mathcal{L}^{n-1}, \{-\}_n).
\]
Notice that for $i\geq 1$ we have 
\[
\mathcal{L}^{(i+1)}(a_1,\dots,a_i,a_{i+1})=[\mathcal{L}^{(i)}(a_1,\dots,a_i),\rho_{a_{i+1}}].
\]
\end{remark}

\begin{proposition}\label{prop:multiplicative-filtration}
The Grothendieck filtration is multiplicative: if $M,N$ and $P$ are graded $A$-modules, composition of morphisms of $k$-vector spaces induces a composition
\[
\circ\colon \Diff_A^{\leq m}(N,P)
\otimes
\Diff_A^{\leq n}(M,N)
\subseteq
\Diff_A^{\leq m+n}(M,P).
\]
In particular, when $M=N$ one sees that $\Diff_A(M,M)$ is a filtered graded algebra.
\end{proposition}

\begin{proof}
If $\min(m,n)<0$ the claim holds trivially, so we may assume $m,n\geq 0$. We proceed by induction on $m+n$. When $m+n=0$ we have $m=n=0$ and the statement reduces to the fact that the composition of $A$-module morphisms is an $A$-module morphism. Next, assume $\Bc\in \Diff_A^{\leq m}(N,P)$ and $\Psi\in \Diff_A^{\leq n}(M,N)$. Then, up to signs that we do not specify since they are not relevant here,
\[
[\Bc\circ \Psi,\rho_a]=\Bc\circ[ \Psi,\rho_a]\pm [\Bc,\rho_a]\circ \Psi.
\]
We have $[ \Psi,\rho_a]\in \Diff_A^{\leq n-1}(M,N)$ and $[\Bc,\rho_a]\in \Diff_A^{\leq m-1}(N,P)$, and the conclusion follows by induction.
\end{proof}
\begin{remark}
When $M=N$, the commutator bracket on $\End_k(M)$ induces a commutator bracket on $\Diff_A(M,M)$. Moreover, if the total Grothendieck order is strictly positive, the commutator bracket lowers it by one: 
\[
[\,,\,]\colon \Diff_A^{\leq m}(M,M)
\otimes
\Diff_A^{\leq n}(M,M)
\subseteq
\Diff_A^{\leq m+n-1}(M,M), \qquad \text{if } m+n>0.
\]
This is an immediate consequence of the Jacobi identity. Indeed, in the same notation as in the proof of Proposition \ref{prop:multiplicative-filtration}, we have:
\[
[[\Bc,\Psi],\rho_a]=[[\Bc,\rho_a],\Psi]\pm [\Bc,[\Psi,\rho_a]].
\]
Since $[\Bc,\rho_a]\in \Diff_A^{\leq m-1}$, $[\Psi,\rho_a]\in \Diff_A^{\leq n-1}$, we inductively have
$[[\Bc,\rho_a],\Psi], [\Bc,[\Psi,\rho_a]],\in \Diff_A^{\leq m+n-2}(M,M)$, and so $[[\Bc,\Psi],\rho_a]\in \Diff_A^{\leq m+n-2}(M,M)$, i.e., $[\Bc,\Psi]\in \Diff_A^{\leq m+n-1}(M,M)$. The base of the induction is the case $m=0$ or $n=0$. It is not restrictive to assume $m>0$ and $n=0$. In this case we find 
\[
[[\Bc,\Psi],\rho_a]=[[\Bc,\rho_a],\Psi]=[\Bc,\rho_a]\circ \Psi\pm \Psi\circ [\Bc,\rho_a],
\]
and this is in $\Diff_A^{\leq m-1}$ by Proposition \ref{prop:multiplicative-filtration}, since  $[\Bc,\rho_a]$ is in $\Diff_A^{\leq {m-1}}$ and $\Psi$ is in $\Diff_A^{\leq 0}$.
As a consequence $\Diff^{\leq 1}_A(M,M)$ is a graded Lie subalgebra of the graded Lie algebra $(\Diff^{\leq 1}_A(M,M), [\,,\,])$.
\end{remark}

\begin{proposition}\label{prop:subquotients}
The Grothendieck filtration is compatible with subquotient constructions: if $M''\subseteq M'\subseteq M$ and $N''\subseteq N'\subseteq N$ are inclusions of graded $A$-modules, $\Hom_k(M^\bullet,N^\bullet)$ denotes the $A$-module of $k$-linear morphisms from $M$ to $N$ mapping $M'$ to $N'$ and $M''$ to $N''$ and
we denote by $\Diff_A^{\leq n}(M^\bullet,N^\bullet)$ the intersection
\[
\Diff_A(M^\bullet,N^\bullet)=\Diff_A^{\leq n}(M,N)\cap \Hom_k(M^\bullet,N^\bullet),
\]
then the natural morphism $\Hom_k(M^\bullet,N^\bullet)\to \Hom_k(M'/M'',N'/N'')$ induces a morphism of filtered $A$-modules
\[
\Diff_A(M^\bullet,N^\bullet)\to \Diff_A(M'/M'',N'/N'').
\]
\end{proposition}
\begin{proof}
    We have to show that, if $\Bc\in \Diff_A^{\leq n}(M^\bullet,N^\bullet)$, then the induced operator $\overline{\Bc}\in \Hom_k(M'/M'',N'/N'')$ is an element in $\Diff_A^{\leq n}(M'/M'',N'/N'')$. We argue by induction on $n$. For $n=0$, the morphism $\Bc$ is an element in $\Hom_A(M^\bullet,N^\bullet)$ and so $\overline{\Bc}\in \Hom_A(M'/M'',N'/N'')=\Diff_A^{\leq 0}(M'/M'',N'/N'')$. If $n>0$, we use the  fact that the multiplication by $a\in A$ on the quotient $M'/M''$ is induced by the multiplication by $a$ on $M$, so that $\rho^{M'.M''}_a=\overline{\rho^M_a}$, and similarly for $N$. Therefore,
    \[
    [\overline{\Bc},\rho_a]=[\overline{\Bc},\overline{\rho_a}]=\overline{[\Bc,\rho_a]},
    \]
    and by the inductive assumption the last term is in $\Diff_A^{\leq n-1}(M'/M'',N'/N'')$. Hence, $\overline{\Bc}\in \Diff_A^{\leq n}(M'/M'',N'/N'')$.
\end{proof}

\begin{corollary}
Let $(M,d_M)$ and $(N,d_N)$ be cochain complexes, and let $\Bc\colon M\to N$ a morphism of complexes. If $\Bc$ is a Grothendieck differential operator of order less than or equal to $n$, then the induced operator in cohomology $H(\Bc)\colon H(M)\to H(N)$ is again  Grothendieck differential operator of order less than or equal to $n$.   
\end{corollary}

\begin{notation}
Let $A$ be a graded commutative algebra. We denote by $\mathfrak{g}_A$ the graded vector space $\mathfrak{g}_A=A[-1]$. If $M$ is a graded $A$-module, the module structure map
\[
\rho\colon A\otimes M \to M
\]
can be equivalently seen as a map
\[
\iota\colon \mathfrak{g}_A\otimes M\to M[-1].
\]
\end{notation}
\begin{remark}\label{rem:from-rho-to-iota}
Since $\iota\colon \mathfrak{g}_A\to \mathrm{End}_k(M)$ has degree $-1$, we have
\[
\deg(\iota_a)=\deg_{\mathfrak{g}_A}(a)-1=\deg_A(a).
\]
Therefore we see that definition \ref{def:diff} can be equivalently rephrased as
\[
\Diff_A^{\leq n}(M,N)
:=
\left\{
\Bc\in\Hom_{\mathbb Z}(M,N)
\mid
[\Bc,\iota_a]\in
\Diff_A^{\leq n-1}(M,N)
\text{ for all }a\in \mathfrak{g}_A
\right\}.
\]
The derived bracket of Remark \ref{rem:derived-bracket} takes the form
\[
\{-\}_{\Bc}\colon \mathfrak{g}_A^{\otimes n}\to \Hom_A(M,N)[\deg\Bc-n].
\]

\end{remark}

\section{The Cartan differential operator algebra}\label{sec:cartan}

Assume now $M$ is a faithful graded $A$-module, i.e., that the map
\[
\rho\colon A\longrightarrow\End_k(M)
\]
is injective, thus identifying $A$ with its image in $\End_k(M)$. Under this identification, since $A$ is graded commutative, we have $A\subseteq \Hom_A(M,M)$, and so $A\subseteq \Diff_A^{\leq 0}(M,M)$. Taking into account Remark \ref{rem:from-rho-to-iota}, we can then give the following definition
\begin{definition}
The \emph{Cartan differential operator algebra} of the faithful
$A$-module $M$ is the filtered algebra
\[
\Cartan_A(M,M)
=
\bigcup_{n}
\Cartan_A^{\leq n}(M,M)
\]
defined inductively by
\[
\Cartan_A^{\leq-1}(M,M)=0;\qquad \Cartan_A^{\leq0}(M,M)=A
\]
and
\[
\Cartan_A^{\leq n}(M,M)
=
\left\{
\Bc\in\End_k(M)
\mid
[\Bc,\iota_a]\in
\Cartan_A^{\leq n-1}(M,M)
\quad\forall a\in \mathfrak{g}_A
\right\}.
\]
\end{definition}
\begin{lemma}
$\Cartan_A(M,M)$ is a filtered subalgebra of $\Diff_A(M,M)$.    
\end{lemma}
\begin{proof}
 Since $\Cartan_A^{\leq 0}(M,M)=A\subseteq \mathrm{End}_A(M)=\Diff^{\leq 0}_A(M,M)$, one inductively sees that $\Cartan_A(M,M)$ is a filtered submodule of $\Diff_A(M,M)$. To see it is a filtered subalgebra, we have to show that the composition of endomorphisms of the $k$-vector space $M$ induces a composition    
 \[
\circ\colon \Cartan_A^{\leq m}(M,M)
\otimes
\Cartan_A^{\leq n}(M,M)
\subseteq
\Cartan_A^{\leq m+n}(M,M).
\]
This is proven inductively verbatim following the argument in the proof of Proposition \ref{prop:multiplicative-filtration}. The basis of the induction is the fact that the composition of two Cartan differential operators of order less than or equal to zero is the multiplication of two elements of $A$ and so it is again an element in $A$.
\end{proof}

\begin{remark}Note that one cannot even define a notion of $\Cartan(M,N)$, contrary to $\Diff(M,N)$.
\end{remark}
\begin{lemma}
    If $\varphi\colon M\to N$ is an isomorphism of graded $A$-modules, then conjugation with $\varphi$ induces an isomorphism of filtered algebras between $\Cartan_A(M,M)$ and $\Cartan_A(N,N)$. 
\end{lemma}
\begin{proof}
 We have to show that $\varphi\circ\Cartan_A^{\leq n}(M,M)\circ \varphi^{-1}\subseteq  \Cartan_A^{\leq n}(N,N)$. The proof is by induction on $n$. When $n=0$, the fact that, by definition, morphisms of $A$-modules commute with multiplications by elements in $A$ gives  $\varphi\circ A\circ \varphi^{-1}=A$. Let now $\Bc\in \Cartan_A^{\leq n}(M,M)$ with $n>0$. Since $[\varphi,\iota_a]=[\varphi^{-1},\iota_a]=0$, we have
 \[
 [\varphi\circ\Bc\circ \varphi^{-1},\iota_a]=\pm \varphi\circ [\Bc,\iota_a]\circ \varphi^{-1}.
 \]
Now, $[\Bc,\iota_a]\in \Cartan_A^{\leq n-1}(M,M)$ and the statement follows by induction.
\end{proof}
\begin{example}\label{example:for-bv}
Let $M$ be a free rank 1 $A$-module. Then the choice of a basis element\footnote{Although it is implied by the defining property of a basis element, let us explicitly remark that a basis element for a graded $A$-module has necessarily degree zero.} $\zeta$ determines a distinguished isomorphism of graded $A$-modules $\varphi_\zeta\colon M\to A$, defined as the unique morphism of $A$-modules from $M$ to $A$ mapping $\zeta$ to $1$. The choice of $\zeta$ then determines an isomorphism of filtered (and graded) algebras 
\[
\Cartan_A(M,M)\xrightarrow{\sim}\Cartan_A(A,A).
\]
Denoting by $\Delta_{\Bc,\zeta}$ the image of $\Bc$ via this isomorphism, the Cartan differential operator $\Delta_{\Bc,\zeta}\colon A\to A$ is characterized by the defining equation
\begin{equation}\label{eq:defining}
\Delta_{\Bc,\zeta}(a)\, \zeta= \Bc(a\zeta).
\end{equation}
\end{example}
\begin{remark}\label{rem:non-canonical}
Since $\Bc$ is generally not $A$-linear, and not even $A^0$-linear, where $A^0\subseteq A$ is the subalgebra of degree zero elements, the operator $\Delta_{\Bc,\zeta}$ generally depends on the choice of the basis element $\zeta$. 
On the other hand, since $\Bc$ is a $k$-linear operator, equation \eqref{eq:defining} implies $\Delta_{\Bc,x\zeta}=\Delta_{\Bc,\zeta}$ for any $x\in k^\times$.
\end{remark}
From Remark \ref{rem:derived-bracket} we immediately have the following characterization.
\begin{lemma}
Under the identification of $\mathfrak{g}_A$ with its image in $\mathrm{End}_A(M)$ via $\iota$, we have
\[
\Cartan_A^{\leq n}(M,M)
=
\left\{
\Bc\in\End_k(M)
\mid
[\cdots[[\Bc,\iota_{a_1}],\iota_{a_2}],\ldots,\iota_{a_n}]\in \mathfrak{g}_A
\quad\forall a_1,\dots, a_n\in \mathfrak{g}_A
\right\}.
\]
Every element $\Bc$ in $\Cartan_A^{\leq n}(M,M)$ defines a derived bracket 
\[
\{-\}_{\Bc}\colon \mathfrak{g}_A^{\otimes n}\to \mathfrak{g}_A[\deg\Bc+1-n]
\]
given by
\[
[\cdots[[\iota_{a_1},\Bc],\iota_{a_2}],\ldots,\iota_{a_n}]=\iota_{\{a_1,\dots,a_n\}_{\Bc}}.
\]
\end{lemma}

\begin{remark}\label{rem:locally-free}
The inclusion $\Cartan_A(M,M)\hookrightarrow \Diff_A(M,M)$ is generally proper: in order less than or equal to zero one has that $A\subseteq \mathrm{End}_A(M)$ is generally a proper inclusion. On the other hand, due to their inductive definition, as soon as $\rho\colon A\xrightarrow{\sim} \mathrm{End}_A(M)$ is an isomorphism, we have $\Diff_A(M,M)=\Cartan_A(M,M)$. This notably happens when $M$ is a free rank 1 $A$-module, or more generally, when $M$ is a finitely generated faithful projective $A$-module of constant local rank $1$.
\end{remark}

\begin{remark}\label{rem:derived-brackets-cartan}
By verbatim adapting Remark \ref{rem:sequence}, we see that a Cartan differential operator of order at most $n$ and degree $k$ on $M$ is equivalent to a sequence of 
morphisms of $k$-vector spaces of degree $k-i$
\[
\lambda^{(i)}\colon \mathfrak{g}_A^{\otimes i}\to \mathrm{End}_k(M)[k-i], \qquad \text{for }i=0,\dots,n
\]
such that
\begin{itemize}
    \item $\lambda^{(i+1)}(a_1,\dots,a_i,a_{i+1}) = [\lambda^{(i)}(a_1,\dots,a_i), \iota_{a_{i+1}}]$ for any $i=0,\dots, n-1$;
    \item $\lambda^{(n)}\colon \mathfrak{g}_A^{\otimes n} \to \mathfrak{g}_A[k+1-n]$.
\end{itemize}
As in Remark \ref{rem:sequence}, we will write $\Bc$ for $\lambda^{(0)}$, $\mathcal{L}^{(i)}$ for $-(-1)^{(\deg_{\mathfrak{g}_A}(a_1)-1)k}\lambda^{(i)}$, for $i=1,\dots,n$ and $\{-\}_n$, or simply $\{-\}$ when the arity is clear from the context, for $\mathcal{L}^{(n)}$, so to have
\[
[[\cdots[[\iota_{a_1},\Bc],\iota_{a_2}],\ldots],\iota_{a_n}]=[\cdots[[\mathcal{L}^{(1)}_{a_1},\iota_{a_2}],\ldots],\iota_{a_n}]=[\cdots[\mathcal{L}^{(2)}_{a_1,a_2},\ldots],\iota_{a_n}]=\cdots=\iota_{\{a_1,\dots,a_n\}_n}.
\]
\end{remark}
\begin{example}\label{ex:order2}
A Cartan differential operator of order at most 2 and degree $k$ on $M$ is a triple $(\Bc,\mathcal{L},\{-,-\})$, where 
\begin{itemize}
\item $\Bc\colon M\to M[k]$ is a morphism of graded $k$-vector spaces;
\item $\mathcal{L}\colon \mathfrak{g}_A\to \mathrm{End}_k(M)[k-1]$  is defined by
\[
\mathcal{L}_a=[\iota_a,\Bc]
\]
for any $a\in \mathfrak{g}_A$;
\item $\{-,-\}\colon \mathfrak{g}_A\otimes \mathfrak{g}_A\to \mathfrak{g}_A[k-1]$ satisfies
\[
[\mathcal{L}_a,\iota_b]=\iota_{\{a,b\}}
\]
for any $a,b\in \mathfrak{g}_A$.
\end{itemize}
Here we see two of the Cartan equations relating the contraction and the Lie derivative (with $\Bc$ playing the role of the de Rham differential). The third one, namely $[\iota_a,\iota_b]=0$, is an immediate consequence of the graded commutativity of $A$. Notice that when $\Bc$ has degree 1, exactly as in the case of the de Rham operator, both $\mathcal{L}$ and the bracket $\{-,-\}$ have degree zero. Moreover, when $\Bc$ has degree 1 (or more generally when it has odd degree), the bracket $\{-,-\}$ is antisymmetric. Indeed, from the graded Jacobi identity we have
\begin{align*}
\iota_{\{b,a\}}&=[[\iota_b,\Bc],\iota_a]=[\iota_b,[\Bc,\iota_a]]+(-1)^{\deg_{\mathfrak{g}_A}(b)}[\Bc,[\iota_b,\iota_a]]\\
&=-(-1)^{\deg_{\mathfrak{g}_A}(a)(\deg_{\mathfrak{g}_A}(b)-1)}
[[\Bc,\iota_a],\iota_b]\\
&=-(-1)^{\deg_{\mathfrak{g}_A}(a)\deg_{\mathfrak{g}_A}(b)}
[[\iota_a,\Bc],\iota_b]\\
&=-(-1)^{\deg_{\mathfrak{g}_A}(a)\deg_{\mathfrak{g}_A}(b)}\iota_{\{a,b\}},
\end{align*}
and the conclusion follows from the injectivity of $\iota$.
\end{example}
{
\begin{definition}
A Cartan differential operator $(\Bc,\mathcal{L},\{-,-\})$ of order at most 2 and degree $1$ on $M$ such that $\Bc^2=0$ will be called a \emph{$BV_0$-operator} on $M$. When the opposite gradings are taken on $A$ and $M$, so that $\Bc$ has degree $-1$, one calls $\Bc$ a $BV_2$- or \emph{Batalin-Vilkovisky operator} on $M$. 
\end{definition}
}
\begin{lemma}\label{lem:BcL}
Let  $(\Bc,\mathcal{L},\{-,-\})$ be a $BV_0$-operator on $M$. Then
\[
[\Bc,\mathcal{L}_a]=0.
\]
\end{lemma}
\begin{proof}
We have
\begin{align*}
[\Bc,\mathcal{L}_a]&=[\Bc,[\iota_a,\Bc]]=[[\Bc,\iota_a],\Bc]-(-1)^{\deg_{\mathfrak{g}_A}(a)}[\iota_a,[\Bc,\Bc]]\\
&=-(-1)^{\deg_{\mathfrak{g}_A}(a)}[\Bc,[\Bc,\iota_a]]=-[\Bc,[\iota_a,\Bc]]=-[\Bc,\mathcal{L}_a].
\end{align*}
since $[\Bc,\Bc]=2\Bc^2=0$. Hence $2[\Bc,\mathcal{L}_a]=0$ and so $[\Bc,\mathcal{L}_a]=0$.
\end{proof}
\begin{lemma}\label{lem:Lie}
     Let  $(\Bc,\mathcal{L},\{-,-\})$ be a $BV_0$-operator on $M$. Then
     \[
     \mathcal{L}_{\{a,b\}}=[\mathcal{L}_a,\mathcal{L}_b].
     \]
\end{lemma}
\begin{proof}
    We have
    \begin{align*}
    \mathcal{L}_{\{a,b\}}&=[\iota_{\{a,b\}},\Bc]\\
    &=[[\mathcal{L}_a,\iota_b],\Bc]\\
    &=[\mathcal{L}_a,[\iota_b,\Bc]]-(-1)^{\deg_{\mathfrak{g}_A}(a)(\deg_{\mathfrak{g}_A}(b)-1)}[\iota_b,[\mathcal{L}_a,\Bc]]\\
    &=[\mathcal{L}_a,\mathcal{L}_b],
    \end{align*}
    where we used Lemma \ref{lem:BcL}.
\end{proof}
\begin{lemma}\label{lem:Lie-algebra}
 Let  $(\Bc,\mathcal{L},\{-,-\})$ be a $BV_0$-operator on $M$. Then, the bracket $\{-,-\}\colon \mathfrak{g}_A\otimes \mathfrak{g}_A\to \mathfrak{g}_A$ defines a Lie algebra structure on $\mathfrak{g}_A$.
 \end{lemma}
\begin{proof}
We have already seen in Example \ref{ex:order2} that  $\{-,-\}$ is antisymmetric, so we are reduced to showing the graded Jacobi identity for it. We have
\begin{align*}
 \iota_{\{a,\{b,c\}\}}&=[\mathcal{L}_a,\iota_{\{b,c\}}]\\
 &=[\mathcal{L}_a,[\mathcal{L}_b,\iota_{c}]]\\
 &=[[\mathcal{L}_a,\mathcal{L}_b],\iota_c]+(-1)^{\deg_{\mathfrak{g}_{A}}(a)\deg_{\mathfrak{g}_A}(b)}[\mathcal{L}_b,[\mathcal{L}_a,\iota_c]]\\
 &=[\mathcal{L}_{\{a,b\}},\iota_c]+(-1)^{\deg_{\mathfrak{g}_{A}}(a)\deg_{\mathfrak{g}_A}(b)}[\mathcal{L}_b,\iota_{\{a,c\}}]\\
 &=\iota_{\{\{a,b\},c\}}+(-1)^{\deg_{\mathfrak{g}_{A}}(a)\deg_{\mathfrak{g}_A}(b)}\iota_{\{b,\{a,c\}\}},
\end{align*}
where we used Lemma \ref{lem:Lie}. Hence
\[
\{a,\{b,c\}\}=\{\{a,b\},c\}+(-1)^{\deg_{\mathfrak{g}_{A}}(a)\deg_{\mathfrak{g}_A}(b)}\{b,\{a,c\}\}.
\]
\end{proof}
A degree 1 operator $\Bc$ on $M$ with $\Bc^2=0$ makes $(M,\Bc)$ a cochain complex. Since $[\Bc,\Bc]=0$, the element $\Bc$ is also a Maurer--Cartan element in the differential graded Lie algebra $(\mathrm{End}_k(M),0,[\,,\,])$, where $0$ is the zero differential and $[\,,\,]$ is the commutator bracket. Hence we can use $\Bc$ to twist the differential of this dgla to get the dgla $(\mathrm{End}_k(M),[\Bc,\,],[\,,\,])$. Lemmas \ref{lem:BcL}--\ref{lem:Lie-algebra} together then give the following.
\begin{proposition}\label{prop:lie-algebra-morphism}
    Let  $(\Bc,\mathcal{L},\{-,-\})$ be a $BV_0$-operator on $M$.  Then 
    \[
    \mathcal{L}\colon (\mathfrak{g}_A,0,\{\,,\,\})\to (\mathrm{End}_k(M),[\Bc,\,],[\,,\,])
    \]
    is a morphism of differential graded Lie algebras.
\end{proposition}
We can look at the Lie bracket $\{-,-\}\colon \mathfrak{g}_A\otimes \mathfrak{g}_A\to \mathfrak{g}_A$ as a degree 1 bracket on the graded commutative algebra $A$. One can then wonder about the compatibility of this bracket with the multiplication of $A$. One has the following.
\begin{proposition}\label{prop:p0}
   Let  $(\Bc,\mathcal{L},\{-,-\})$ be a  $BV_0$-operator on $M$. Then $(A,\cdot, \{\,,\,\})$ is a $P_0$-algebra, i.e., a Poisson algebra with a bracket of degree $1$.\footnote{Generally, a $P_2$-algebra is a Poisson algebra with a bracket of degree $1=k$ so that, for instance, ordinary Poisson algebras are $P_1$-algebras and Gerstenhaber algebras are $P_2$-algebras.}
   
\end{proposition}
\begin{proof}
    We have
    \begin{align*}
    \iota_{\{a,bc\}}&=[\mathcal{L}_a,\iota_{bc}]=[\mathcal{L}_a,\iota_b\iota_c]=[\mathcal{L}_a,\iota_b]\iota_c+(-1)^{\deg_{\mathfrak{g}_A}(a)(\deg_{\mathfrak{g}_A}(b)-1)}\iota_b[\mathcal{L}_a,\iota_c]\\
    &=\iota_{\{a,b\}}\iota_c+(-1)^{\deg_{\mathfrak{g}_A}(a)(\deg_{\mathfrak{g}_A}(b)-1)}
    \iota_b\iota_{\{a,c\}}\\
    &=\iota_{\{a,b\}c}+(-1)^{(\deg_{A}(a)+1)\deg_A(b)}
    \iota_{b\{a,c\}}\\
    &=\iota_{\{a,b\}c+(-1)^{(\deg_{A}(a)+1)\deg_A(b)}b\{a,c\}}
    \end{align*}
     Hence,
\begin{equation}\label{eq:poisson}
{\{a,bc\}}=\{a,b\}c+(-1)^{(\deg_{A}(a)+1)\deg_A(b)}b\{a,c\}.
\end{equation} 
\end{proof}
\begin{corollary}
    Let  $(\Bc,\mathcal{L},\{-,-\})$ be a $BV_0$-operator on $M$, and let $\mathcal{G}_A$ be the graded commutative algebra obtained from $A$ by taking the opposite grading. Then  $(\mathcal{G}_A,\cdot, \{\,,\,\})$ is a $P_2$-algebra, i.e., a Gerstenhaber algebra.
\end{corollary}
\begin{definition}\label{def.P0-module}
Let $A$ be a graded commutative algebra. A $P_0$-module for $A$ is a pair $(M,\Bc)$ where $M$ is a free $A$-module of rank 1 and $\Bc$ is a differential operator of order at most 2 and degree $1$ on $M$ such that $\Bc^2=0$.  
\end{definition}
{
\begin{remark}
Since for a free rank 1 $A$-module $M$ we have $\Diff_A(M,M)=\Cartan_A(M,M)$, one can equivalently define a  $P_0$-module for $A$ as a pair $(M,\Bc)$ where $M$ is a free $A$-module of rank 1 and $\Bc$ is a $BV_0$-operator on $M$.  
\end{remark}
}
\begin{lemma}\label{lem:p0-module}
A $P_0$-module for $A$ endows $A$ with the structure of $P_0$-algebra.
\end{lemma}
\begin{proof}
  Since $M$ is a free $A$-module of rank 1, we have $\Diff_A^{\leq 2}(M,M)=\Cartan_A^{\leq 2}(M,M)$. The conclusion then follows from Proposition \ref{prop:p0}.  
\end{proof}
\begin{definition}\label{def:bv}
 A $BV_0$-algebra is a graded commutative algebra $A$ endowed with a second order differential operator $\Delta\in \Diff^{\leq 2}(A,A)$ of degree 1, such that $\Delta^2=0$ and $\Delta(1)=0$. $BV_2$-algebra (or simply a $BV$-algebra or a \emph{Batalin-Vilkovisky algebra}) is a $BV_0$-algebra with the opposite grading.   
\end{definition}
\begin{remark}\label{rem:bv-to-p0}
In Definition \ref{def:bv}, one looks at $A$ as a module over itself. Such a module is clearly free of rank 1, so $\Delta\in \Cartan^{\leq 2}(A,A)$, and all of the above considerations apply to $(A,\Delta)$.
Therefore, we can equivalently say that a $BV_0$-algebra is a graded commutative algebra $A$ endowed with a $BV_0$-operator $\Delta$ on $A$ seen as an $A$-module, such that $\Delta(1)=0$. 
In particular, every $BV_0$ algebra is a $P_0$ algebra and every Batalin-Vilkovisky algebra is a Gerstenhaber algebra.  
\end{remark}
\begin{lemma}\label{lem:bv-bracket}
Let $(A,\Delta)$ be a $BV_0$-algebra, and let $\{\,,\,\}$ be its $P_0$-bracket. Then we have
\begin{equation}\label{eq:Delta-and-bracket}
\Delta(ab)=\Delta(a)b+(-1)^{\deg_A(a)}a\Delta(b)-(-1)^{\deg_A(a)}\{a,b\}, 
\end{equation}
for any $a,b\in A$.
\end{lemma}
\begin{proof}
 Let $\mathcal{L}\colon (\mathfrak{g}_A,0,\{\,,\,\})\to (\mathrm{End}_k(A), [\Delta,-], [\,,\,])$ the dgla morphism associated to $\Delta$. We have $[\mathcal{L}_a,\iota_b]=\iota_{\{a,b\}}$. Applying this to the element $1\in A$ we obtain
 \[
 [[\iota_a,\Delta],\iota_b](1)=\{a,b\}.
 \]
 The left hand side is
 \begin{align*}
& \left((\iota_a\circ\Delta-(-1)^{\deg_A(a)}\Delta\circ\iota_a)\circ \iota_b-(-1)^{(\deg_A(a)+1)\deg_A(b)}\iota_b\circ (\iota_a\circ\Delta-(-1)^{\deg_A(a)}\Delta\circ\iota_a)\right)(1)\\
 &=
 a\Delta(b)-(-1)^{\deg_A(a)}\Delta(ab)-(-1)^{(\deg_A(a)+1)(\deg_A(b)+1)}b\Delta(a)\\
 &=
 a\Delta(b)-(-1)^{\deg_A(a)}\Delta(ab)+(-1)^{\deg_A(a)}\Delta(a)b.
 \end{align*}
Thus we obtain
\[
(-1)^{\deg_A(a)}\Delta(ab)=a\Delta(b)+(-1)^{\deg_A(a)}\Delta(a)b-\{a,b\}, 
\]
i.e.,
\[
\Delta(ab)=\Delta(a)b+(-1)^{\deg_A(a)}a\Delta(b)-(-1)^{\deg_A(a)}\{a,b\}. 
\]
\end{proof}
We have the following converse to Lemma \ref{lem:bv-bracket}, thus providing an equivalent definition of a $BV_0$-algebra (and so of a Batalin-Vilkovisky algebra).
\begin{lemma}
Let $(A,\cdot, \{\,,\,\})$ be a $P_0$-algebra, and let $\Delta\colon A\to A[1]$ a degree 1 linear operator such that $\Delta^2=0$, $\Delta(1)=0$, and 
\begin{equation}\label{eq:bv-bracket}
\Delta(ab)=\Delta(a)b+(-1)^{\deg_A(a)}a\Delta(b)-(-1)^{\deg_A(a)}\{a,b\}, 
\end{equation}
for any $a,b\in A$. Then $(A,\Delta)$ is a $BV_0$-algebra.
\end{lemma}
\begin{proof}
We have to show that $\Delta$ is a second order operator. By example \ref{ex:order2}, this is equivalent to showing that    
$[\mathcal{L}_a,\iota_b]=\iota_{\{a,b\}}$,
with $\mathcal{L}_a=[\iota_a,\Delta]$. By the same computation as in the proof of Lemma \ref{lem:bv-bracket}, and writing $|a|$ for $\deg_A(a)$, we have
\[
[\mathcal{L}_a,\iota_b](c)=a\Delta(bc)-(-1)^{|a|}\Delta(abc)-(-1)^{|b|}ab\Delta(c)+(-1)^{|ab|+|a|+|b|}b\Delta(ac).
\]
We now expand the  second and the fourth term on the right hand side by using \eqref{eq:bv-bracket}. We have, for the second term:
\[
\Delta(abc)=
    \Delta(a(bc))\\
    =  \Delta(a)bc + (-1)^{|a|} a\Delta(bc) - (-1)^{|a|} \{a,bc\}, 
\]
and for the fourth term:
\[
b\Delta(ac)=(-1)^{(|a|+1)|b|}\Delta(a)bc+(-1)^{|a|(|b|+1)}ab\Delta(c)-(-1)^{|a|}b\{a,c\}.
\]
Therefore,
\[
[\mathcal{L}_a,\iota_b](c)=\{a,bc\}
-(-1)^{|ab|+|b|}b\{a,c\}=\{a,b\}c=\iota_{\{a,b\}}(c),
\]
where we used \eqref{eq:poisson}.
\end{proof}

\begin{definition}
Let $A$ be a graded commutative algebra. A $BV_0$-module for $A$ is a triple $(M,\Bc,\zeta)$ where $(M,\Bc)$ is a $P_0$-module for $A$ and $\zeta$ is a basis element for $M$ such that $\Bc(\zeta)=0$.  
\end{definition}
\begin{proposition}\label{prop:bv0-module}
 A $BV_0$-module for $A$ induces a structure of $BV_0$-algebra on $A$.    
\end{proposition}
\begin{proof}
Immediate from Example \ref{example:for-bv}. As in the proof of Lemma \ref{lem:p0-module}, the fact that $M$ is free of rank 1 gives $\Bc\in \Cartan_A(M,M)$. Since the isomorphism
\begin{equation}\label{eq:iso-zeta}
\Cartan_A(M,M)\xrightarrow{\sim}\Cartan_A(A,A).
\end{equation}
induced by $\zeta$
is an isomorphism of filtered and graded algebras and $\Cartan_A(A,A)$ is a filtered and graded subalgebra of $\Diff_A(A,A)$  we have that the image $\Delta_{\Bc,\zeta}$ of $\Bc$ via the isomorphism \eqref{eq:iso-zeta} is an element in $\Diff^{\leq 2}(A,A)$ of degree 1. Moreover, $\Bc^2=0$ implies $\Delta_{\Bc,\zeta}^2=0$, and the defining equation  $\Delta_{\Bc,\zeta}(a)\, \zeta= \Bc(a\zeta)$ implies $\Delta_{\Bc,\zeta}(1)\, \zeta= \Bc(\zeta)=0$ and so $\Delta_{\Bc,\zeta}(1)=0$.
\end{proof}
\begin{remark}\label{laplacian:basis-choice}
As  noticed in Remark \ref{rem:non-canonical}, the operator $\Delta_{\Bc,\zeta}$ depends on the choice of the basis element $\zeta$. Thanks to the condition $\Bc(\zeta)=0$, however, this dependence is very well controlled. Namely, if $\zeta'$ is another basis element with $\Bc(\zeta')=0$, then we have $\zeta'=x\zeta$ for some $x\in (A^0)^\times$ and 
\[
\mathcal{L}_x\zeta=[\rho_x,\Bc](\zeta)=x\Bc(\zeta)-\Bc(\zeta')=0.
\]
It follows that, for any $a\in A$, 
\[
\iota_{\{x,a\}}\zeta=\mathcal{L}_x(a\zeta).
\]
Therefore equation \eqref{eq:defining} for $\zeta'$ gives
\[
\Delta_{\Bc,\zeta'}(a)\,x\zeta=\Bc(xa\zeta)=x\Bc(a\zeta)-\mathcal{L}_x(a\zeta)=x\Delta_{\Bc,\zeta}(a)\zeta-{\{x,a\}}\zeta,
\]
hence
\[
\Delta_{\Bc,\zeta'}=\Delta_{\Bc,\zeta}-x^{-1}\{x,-\}.
\]
\end{remark}
\begin{remark}\label{rem:canonical}
If $(M,\Bc,\zeta)$ is a $BV_0$-module for $A$, then we have \emph{two} $P_0$-algebra structures induced by 
$(M,\Bc,\zeta)$ on $A$. The first one is the one obtained forgetting the basis element $\zeta$ and looking just at the $P_0$-module $(M,\Bc)$. By Lemma \ref{lem:p0-module} this induces a $P_0$-algebra structure on $A$. Let us denote the Poisson bracket of this structure by the symbol $\{\,,\,\}$. The second one consists in considering the $BV_0$-algebra $(A,\cdot, \Delta_{\Bc,\zeta})$ from Proposition \ref{prop:bv0-module}, and consider the underlying $P_0$-algebra (see Remark \ref{rem:bv-to-p0}). Let us denote the Poisson bracket of this structure by the symbol $\{\,,\,\}_\zeta$.  As the notation emphasizes, the first bracket does not depend on the choice of the basis element $\zeta$. So it is, in a sense, more canonical than $\{\,,\,\}_\zeta$. But actually something remarkable happens here: the bracket $\{\,,\,\}_\zeta$ is independent of $\zeta$ and, for any choice of $\zeta$ it coincides with the canonical bracket $\{\,,\,\}$. In other words, the enhancement of a $P_0$-module to a $BV_0$-module corresponds to an enhancement of the associated $P_0$-algebra to a $BV_0$-algebra. The fact that the $P_0$-bracket remains unchanged as the basis element $\zeta$ varies, while the $BV_0$-operator $\Delta_{\Bc,\zeta}$ does generally depend on the choice of $\zeta$ can be seen as a general instance of the phenomenon commonly encountered in Physics, where in the Batalin--Vilkovisky quantization procedure the Gerstenhaber bracket is typically canonical, whereas the Batalin--Vilkovisky laplacian depends on additional choices. The proof of the equality $\{\,,\,\}_\zeta=\{\,,\,\}$ is given in the following Lemma. 
\end{remark}

\begin{lemma}
 In the setting and notation of Remark \ref{rem:canonical}, one has  $\{\,,\,\}_\zeta=\{\,,\,\}$.  
\end{lemma}
\begin{proof}
By Lemma \ref{lem:bv-bracket}, we have 
\[
\Delta_{\Bc,\zeta}(ab)=\Delta_{\Bc,\zeta}(a)b+(-1)^{\deg_A(a)}a\Delta_{\Bc,\zeta}(b)-(-1)^{\deg_A(a)}\{a,b\}_\zeta.
\]
So in order to prove the statement it will suffice to show that
\[
\Delta_{\Bc,\zeta}(ab)-\Delta_{\Bc,\zeta}(a)b-(-1)^{\deg_A(a)}a\Delta_{\Bc,\zeta}(b)=-(-1)^{\deg_A(a)}\{a,b\}.
\]
Let us multiply the left hand side by $\zeta$ and apply the defining property of $\Delta_{\Bc,\zeta}$, i.e., equation \eqref{eq:defining}. Reasoning as in the proof of Lemma \ref{lem:bv-bracket}, and using $\Bc(\zeta)=0$, we have
\begin{align*}
&\left(\Delta_{\Bc,\zeta}(ab)-\Delta_{\Bc,\zeta}(a)b-(-1)^{\deg_A(a)}a\Delta_{\Bc,\zeta}(b)\right)\zeta\\
&\qquad\qquad=
\Delta_{\Bc,\zeta}(ab)\zeta -(-1)^{(\deg_A(a)+1)\deg_A(b)}b\Delta_{\Bc,\zeta}(a)\zeta -(-1)^{\deg_A(a)}a\Delta_{\Bc,\zeta}(b)\zeta\\
&\qquad\qquad=
\Bc(ab\zeta) -(-1)^{(\deg_A(a)+1)\deg_A(b)}b\Bc(a\zeta) -(-1)^{\deg_A(a)}a\Bc(b\zeta)\\
&\qquad\qquad\qquad\qquad -(-1)^{(\deg_A(a)+1)(\deg_A(b)+1)} ab\Bc(\zeta)\\
&\qquad\qquad=
\left(\Bc\circ\iota_a\circ\iota_b  
-(-1)^{(\deg_A(a)+1)\deg_A(b)}\iota_b \circ\Bc\circ\iota_a 
\right.\\
&\qquad\qquad\qquad\qquad \left.
-
(-1)^{\deg_A(a)}\iota_a\circ\Bc\circ\iota_b -(-1)^{(\deg_A(a)+1)(\deg_A(b)+1)} \iota_a\circ\iota_b\circ\Bc
\right)(\zeta)\\
&\qquad\qquad=
-(-1)^{\deg_A(a)}[[\iota_a,\Bc],\iota_b](\zeta)\\
&\qquad\qquad=
-(-1)^{\deg_A(a)}\{a,b\}\zeta.
\end{align*}
\end{proof}
{
Let $(A,\Delta)$ be a $BV_0$-algebra. We can look at the operator $\Delta\colon A\to A[1]$ as a linear operator $\Delta\colon\mathfrak{g}_A\to \mathfrak{g}_A[1]$. Then we have the following.
\begin{lemma}\label{lem:dgal-from-Delta}
Let $(A,\Delta)$ be a $BV_0$-algebra, and let $(\Delta,\mathcal{L},\{\,,\})$ be the triple associated with $\Delta$. Then $(\mathfrak{g}_A,\Delta,\{\,,\,\})$ is a differential graded Lie algebra.
\end{lemma}
\begin{proof}
We already know that $(\mathfrak{g}_A,\{\,,\,\})$ is a graded Lie algebra and that $\Delta$ is a differential, so it remains only to prove that $\Delta$ is a derivation of the bracket, i.e., that we have
\[
\Delta\{a,b\}=\{\Delta(a),b\}+(-1)^{\deg_{\mathfrak{g}_A}(a)}\{a,\Delta(b)\}.
\]
By applying $\Delta$ to \eqref{eq:Delta-and-bracket} we obtain
\[
0=\Delta^2(ab)
=\Delta\bigl(\Delta(a)b\bigr)
+(-1)^{\deg_A(a)}\Delta\bigl(a\Delta(b)\bigr)
-(-1)^{\deg_A(a)}\Delta\{a,b\}.
\]
Now we apply \eqref{eq:Delta-and-bracket} again to expand the first and the second term on the right hand side. Since $\deg_A(\Delta(a))=\deg_A(a)+1$, we find
\begin{align*}
\Delta\bigl(\Delta(a)b\bigr)
&=
\Delta^2(a)b
+(-1)^{\deg_A(a)+1}\Delta(a)\Delta(b)
-(-1)^{\deg_A(a)+1}\{\Delta(a),b\}\\
&=
(-1)^{\deg_A(a)+1}\Delta(a)\Delta(b)
+(-1)^{\deg_A(a)}\{\Delta(a),b\}
\end{align*}
and
\begin{align*}
(-1)^{\deg_A(a)}\Delta\bigl(a\Delta(b)\bigr)
&=
(-1)^{\deg_A(a)}
\left(
\Delta(a)\Delta(b)+(-1)^{\deg_A(a)}a\Delta^2(b)
-(-1)^{\deg_A(a)}\{a,\Delta(b)\}
\right)\\
&=
(-1)^{\deg_A(a)}\Delta(a)\Delta(b)
-\{a,\Delta(b)\}.
\end{align*}
Therefore,
\[
0=(-1)^{\deg_A(a)}\{\Delta(a),b\}-\{a,\Delta(b)\}-(-1)^{\deg_A(a)}\Delta\{a,b\}
\]
i.e.,
\[
\Delta\{a,b\}=\{\Delta(a),b\}+(-1)^{\deg_A(a)+1}\{a,\Delta(b)\}=\{\Delta(a),b\}+(-1)^{\deg_{\mathfrak{g}_A}(a)}\{a,\Delta(b)\}.
\]
\end{proof}
\begin{remark}
    Let $(A,\Delta)$ be a $BV_0$-algebra, and let $(\Delta,\mathcal{L},\{\,,\})$ be the triple associated with $\Delta$. Then, in the notation of Remark \ref{rem:derived-brackets-cartan}, the binary bracket $\{\,,\}_2=\{\,,\,\}$  is defined by the equation $\{a,b\}_2=\mathcal{L}^{(2)}_{a,b}(1)$
    i.e., 
    \[
    \{a,b\}_2=[[\iota_a,\Delta],\iota_b](1).
    \]
    Up to the sign $(-1)^{\deg_{\mathfrak{g}_A}(a)}$, this is the same as $\lambda^{(2)}_{a,b}(1)$.
    By the same construction we can define a unary bracket $\{\,\}_1$ by the equation $\{a\}_1=\lambda^{(1)}_a(1)$, i.e.,
    \[
    \{a\}_1=[\Delta,\iota](1)
    \]
    and a $0$-ary bracket $\{\,\}_0$ by the equation $\{\,\}_0=\lambda^{(0)}(1)$, i.e.
    \[
    \{\,\}_0=\Delta(1).
    \]
    From the last equation we immediately see that $\{\,\}_0=0$. Concerning the unary bracket $\{\,\}_1$, we have
    \[
    \{a\}_1=\Delta(\iota_a(1))-(1)^{\deg_A(a)}\iota_a(\Delta(1))=\Delta(a),
    \]
    i.e., the unary bracket $\{\,\}$ coincides with the $BV_0$-operator $\Delta$. From Lemma \ref{lem:dgal-from-Delta} we then see that $(\mathfrak{g}_A,\{\,\}_1, \{\,,\,\}_2)$ is a differential graded Lie algebra. This can be seen as an $L_\infty$ algebra with nontrivial brackets only in arity $\leq 2$. One may notice that $\leq 2$ is also the order of the Cartan differential operator $\Delta$. This is no coincidence: more generally, if the order of a degree 1 Cartan differential operator $\Delta\colon A\to A[1]$ with $\Delta^2=0$ is $\leq k$ instead of being $\le 2$, the derived bracket construction
    \[
    \{a_1,\dots,a_k\}_k=\pm\lambda^{(k)}_{a_1,\dots,a_k}(1),
    \]
    with the appropriate choice of signs, produces an $L_\infty$-algebra with nontrivial brackets only in arity $\leq k$. This construction is originally due to Kosmann-Schwarzbach \cite{Kosmann-Schwarzbach-Poisson2Gerstenhaber} and Kravchenko \cite{kravchenko1999}, and then generalized by Voronov \cite{VORONOV2005133} and Bandiera \cite{Bandiera_2015}. In particular, the description of derived brackets given here is taken from \cite[Example 4.4.]{VORONOV2005133}. The terminology homotopy $BV_0$ algebras (and so  \emph{homotopy Batalin-Vilkovisky algebras} when the opposite grading is adopted) to denote the $L_\infty$-algebras obtained this way is proposed in \cite{kravchenko1999}.
\end{remark}

}

\section{Batalin--Vilkovisky algebras from Lie algebroids}
A classical source of examples of $BV_0$ or Batalin-Vilkovisky algebras is provided by Lie algebroids endowed with volume forms. As particular cases, one has the Batalin-Vilkovisky algebras given by divergence operators on smooth manifolds, by holomorphic divergence operators on Calabi-Yau manifolds, and by Chevalley-Eilenberg homological differential in the Chevalley-Eilenberg chain complex of a Lie algebra. We present these three examples first, before giving the general construction for Lie algebroids. Even more generally, as shown by Ping Xu in \cite{Ping}, the construction can be carried to Lie algebroids endowed with a flat connection on their determinant bundle\footnote{Such a flat connection always exists.}.

\begin{example}[Divergence in manifolds with volume forms]
   A notable $P_0$ algebra is given by multivector fields, i.e., sections of the exterior algebra of the tangent bundle, on a smooth manifold $\mathbf{M}$: 
   \[
   (\mathfrak{X}^\bullet (\mathbf{M}), \wedge, [\,,\,]_{SN}), 
   \]
   where $[\,,\,]_{SN}$ denotes the Schouten--Nijenhuis bracket and the grading is given by
   \[
   \mathfrak{X}^j(\mathbf{M})=\Gamma(\mathbf{M};\wedge^{-j}T\mathbf{M}).
   \]
   The graded vector space $\mathcal{E}^\bullet(\mathbf{M})[\dim \mathbf{M}]$ of differential forms on $\mathbb{M}$, shifted in degree by the dimension of $\mathbf{M}$ is a free rank 1 $\mathfrak{X}^\bullet (\mathbf{M})$-module with the action given by contraction. A basis element for this module is given by a volume form $\nu$. Classical Cartan formulas tell us that the de Rham operator $d\colon \mathcal{E}^\bullet(\mathbf{M})[\dim \mathbf{M}]\to \mathcal{E}^\bullet(\mathbf{M})[\dim \mathbf{M}+1]$ is a differential operator of order at most 2 for this module. Since clearly $d$ has degree 1 and $d^2=0$, we have that $d$ is a $BV_0$-operator on $\mathcal{E}^\bullet(\mathbf{M})[\dim \mathbf{M}]$. Since a volume form $\nu$ is a top degree form, we have $d\nu=0$, hence the triple $(\mathcal{E}^\bullet(\mathbf{M})[\dim \mathbf{M}],d,\nu)$ is a $BV_0$-module for $(\mathfrak{X}^\bullet (\mathbf{M}), \wedge, [\,,\,]_{SN})$ and by Proposition \ref{prop:bv0-module} we have an induced $BV_0$-algebra structure on $\mathfrak{X}^\bullet (\mathbf{M})$ enhancing the $P_0$-algebra structure. One easily checks that the induced $BV_0$-operator $\Delta_\nu$ on multivector fields is the $\nu$-divergence operator $\mathrm{div}_\nu$. With the opposite grading we obtain the $\nu$-divergence Batalin-Vilkovisky algebra structure on multivector fields.
    \label{ex.real}
\end{example}
\begin{remark}
Example \ref{ex.real} is illustrative of the fact that, given a linear operator $\Bc\colon M\to M$, its order as a differential operator depends on the $A$-module structure one considers on $M$. For instance, if we look at the graded vector space $\mathcal{E}^\bullet(\mathbf{M})[\dim \mathbf{M}]$ as a $\mathbb{R}$-module, then the de Rham differential has order $\leq 0$. It has order $\leq 1$ if we look at  $\mathcal{E}^\bullet(\mathbf{M})[\dim \mathbf{M}]$ as a ${C}^\infty(\mathbf{M})$-module, and it has order $\leq 2$ if we look at it as a $\mathfrak{X}^\bullet (\mathbf{M})$-module.
\end{remark}

\begin{remark} Via the isomorphism $\varphi_\nu$ of $\mathfrak{X}^\bullet (\mathbf{M})$-modules between $\mathfrak{X}^\bullet (\mathbf{M})$ and $\mathcal{E}^\bullet(\mathbf{M})[\dim \mathbf{M}]$ induced by a volume form $\nu$, linear functionals on the space of differential forms become linear functionals on the space of multivector fields. If we denote by $\Phi_{\mathbf{N};\nu}\colon \mathfrak{X}^\bullet(\mathbf{M})\to \mathbb{R}[\dim \mathbf{M}-\dim \mathbf{N}]$ the linear functional corresponding to integration over a closed oriented submanifold $\mathbf{N}\hookrightarrow\mathbf{M}$, Stokes' theorem implies that, whenever a multivector field $X$ is $\nu$-divergence free, i.e., $\mathrm{div}_\nu(X)=0$, the value $\Phi_{\mathbf{N};\nu}(X)$ depends only on the homology class of $\mathbf{N}$. Moreover, if $X$ is $\mathrm{div}_\nu$-exact, then $\Phi_{\mathbf{N};\nu}(X)=0$. As a particular case, one can consider the volume form $\mathrm{vol}_g$ induced by a Riemannian metric $g$ on $\mathbf{M}$, and take $\mathbf{N}$ of codimension 1. In this case, the resulting functional on vector fields is the flux through the hypersurface $\mathbf{N}$, and one recovers Gau\ss' law as Stokes' theorem read through the isomorphism $\varphi_\nu$.
\label{ex:complex}
\end{remark}

\begin{example}[Holomorphic divergence on Calabi-Yau manifolds]
By taking $\mathbf{M}$ to be a complex manifold and considering the $P_0$-algebra of holomorphic multivector fields
\[(\mathfrak{X}_{\text{hol}}^\bullet(\mathbf{M}), \wedge, [\,,\,]_{SN}), 
   \]
one obtains a holomorphic version of Example \ref{ex.real}. The relevant $\mathfrak{X}_{\text{hol}}^\bullet(\mathbf{M})$-module will be the module $\Omega^\bullet(\mathbf{M})[\dim_{\mathbb{C}}\mathbf{M}]$ of holomorphic differential forms on $\mathbf{M}$ (with degree shifted by the complex dimension of $\mathbf{M}$), and the $BV_0$-operator will be the $\partial$ operator. 
Notice however, that in this case the $\mathfrak{X}_{\text{hol}}^\bullet(\mathbf{M})$-module $\Omega^\bullet(\mathbf{M})[\dim_{\mathbb{C}}\mathbf{M}]$ will generally not be free of rank 1. In order for this to happen, indeed, one needs $\mathbf{M}$ to admit a holomorphic volume form $\nu$, i.e., to be a Calabi-Yau manifold. When this happens, $\partial$ operator is translated to the holomorphic divergence operator $\Delta_\nu$ and one obtains the  Batalin-Vilkovisky algebra of holomorphic multivector fields on Calabi-Yau manifolds \cite{Tian,Barannikov_1998}.
\end{example}

\begin{example}[Chevalley--Eilenberg homological differential]\footnote{We would like to thank Ján Pulmann for leading us here.}
Let $\mathfrak{g}$ be a finite dimensional Lie algebra, and let $\mathrm{CE}^\bullet_{\mathrm{Ho}}(\mathfrak{g})$ be the Chevalley-Eilenberg chain complex of $\mathfrak{g}$, defined by 
\[
\mathrm{CE}^j_{\mathrm{Ho}}(\mathfrak{g})=\wedge^{-j}\mathfrak{g}.
\] 
Endowing $\mathrm{CE}^k_{\mathrm{Ho}}(\mathfrak{g})$ with the product given by wedge product and with 
the Schouten--Nijenhuis-type bracket obtained by extending the Lie bracket of $\mathfrak{g}$ by the graded Leibniz rule, one obtains a $P_0$ algebra. A natural module for this Poisson algebra is given by the graded vector space underlying the shifted Chevalley-Eilenberg cochain complex $\mathrm{CE}^\bullet_{\mathrm{Coh}}(\mathfrak{g})[\dim \mathfrak{g}]$, where
\[
\mathrm{CE}^j_{\mathrm{Coh}}(\mathfrak{g},k)=\mathrm{Hom}_k(\mathrm{CE}^{-j}_{\mathrm{Ho}}(\mathfrak{g}),k),
\]
and where the module structure is given by contraction. The 
Chevalley-Eilenberg cochain complex differential has the form
\[
(d_{\mathrm{CE}}\eta)(\gamma_1\wedge\cdots\wedge\gamma_{k+1})=
\sum_{i<j}\pm \eta([\gamma_i,\gamma_j]^{}_{\mathfrak g}\wedge
\gamma_1\wedge\cdots
\wedge\widehat{\gamma_i}\wedge\cdots\wedge\widehat{\gamma_j}\wedge
\cdots\wedge\gamma_{k+1}^{})
\]
and the same computation leading to classical Cartan identities in Example \ref{ex.real} shows that $d_{\mathrm{CE}}$ is a second order operator, and so a $BV_0$-operator, for the $\mathrm{CE}^\bullet_{\mathrm{Ho}}(\mathfrak{g})$-module structure on $\mathrm{CE}^\bullet_{\mathrm{Coh}}(\mathfrak{g},k)[\dim \mathfrak{g}]$. This module is free of rank 1, with a basis element provided by a nonzero element $\nu\in \mathrm{Hom}_k(\wedge^{\dim\mathfrak{g}}\mathfrak{g},k)$. Such a $\nu$ is automatically $d_{\mathrm{CE}}$-closed, so the choice of $\nu$ translates $d_{\mathrm{CE}}$ to a $BV_0$-operator $\Delta_\nu$ on $\mathrm{CE}^\bullet_{\mathrm{Ho}}(\mathfrak{g})$ enhancing its $P_0$-structure to a $BV_0$-structure. Since $\dim_{k}\mathrm{Hom}_{k}(\wedge^{\dim\mathfrak{g}}\mathfrak{g},k)=1$, the element $\nu$ is uniquely determined up to a nonzero scalar factor. Hence, as remarked in Example \ref{example:for-bv}, $\Delta_\nu$ is actually independent of the choice of $\nu$ and so it defines a canonical $BV_0$-operator $\Delta$ on  $\mathrm{CE}^\bullet_{\mathrm{Ho}}(\mathfrak{g})$. The operator $\Delta$ is in particular a canonical differential and, surprise, it is the homological Chevalley--Eilenberg differential 
\[
\delta_{\mathrm{CE}}(\gamma_1\wedge\cdots\wedge\gamma_{k+1})=\sum_{i<j}\pm [\gamma_i,\gamma_j]^{}_{\mathfrak g}\wedge
\gamma_1\wedge\cdots
\wedge\widehat{\gamma_i}\wedge\cdots\wedge\widehat{\gamma_j}\wedge
\cdots\wedge\gamma_{k+1}.
\]
This way one obtains the canonical $BV_0$-structure $(\mathrm{CE}^\bullet_{\mathrm{Ho}}(\mathfrak{g}),\wedge, \delta_{\mathrm{CE}})$ on the  Chevalley--Eilenberg chain complex of a Lie algebra, see \cite{Gwilliam}.
\label{CE-example}
\end{example}

\begin{example}[Lie algebroids with volume forms]\footnote{We thank warmly Miquel Cueca for pointing this to us.} The three examples described so far are particular instances of a general construction involving a Lie algebroid endowed with a volume form.
To fix notation, let $(\pi, a, [\,,\,]_{\mathbf{E}})$ be a Lie algebroid, where $\pi\colon \mathbf{E}\to \mathbf{M}$ is a finite rank smooth (resp. holomorphic) vector bundle over a smooth (resp. holomorphic) manifold $\mathbf{M}$, $a\colon \mathbf{E}\to T\mathbf{M}$ is the anchor map and $[\,,\,]_\mathbf{E}\colon \mathcal{E}\otimes \mathcal{E}\to \mathcal{E}$ is the Lie bracket of the algebroid, where $\mathcal{E}$ denotes the sheaf of smooth (resp. holomorphic) sections of $\mathbf{E}$. The algebra $\mathrm{CE}^\bullet_{\mathrm{Ho}}({\mathbf{E}})$, with $\mathrm{CE}_{\mathrm{Ho}}^j({\mathbf{E}})=\Gamma(\mathbf{M},\wedge^{-j}\mathcal{E})$, endowed with the wedge product and the Schouten--Nijenhuis-type bracket obtained by extending the Lie bracket of $\mathcal{E}$ by the graded Leibniz rule, is a $P_0$-algebra. A natural module for this algebra, with module action given by contraction, is given by $\mathrm{CE}_{\mathrm{Coh}}^\bullet(\mathbf{E})[\mathrm{rk} \mathbf{E}]$, where $\mathrm{CE}_{\mathrm{Coh}}^j(\mathbf{E})=\Gamma(\mathbf{M},\wedge^j\mathcal{E}^*)$. Here $\mathcal{E}^*$ denotes the sheaf of sections of $\mathbf{E}^*$. If $\mathbf{E}$ admits a volume form, i.e., a nowhere zero element $\nu\in \mathrm{CE}_{\mathrm{Coh}}^{\mathrm{rk} \mathbf{E}}(\mathbf{E})$, then $\mathrm{CE}_{\mathrm{Coh}}^\bullet(\mathbf{E})[\mathrm{rk} \mathbf{E}]$ is a free rank 1 $\mathrm{CE}_{\mathrm{Ho}}^\bullet(\mathbf{E})$-module, with basis element provided by $\nu$. The graded vector space $\mathrm{CE}_{\mathrm{Coh}}^\bullet(\mathbf{E})$ carries a natural graded commutative algebra structure, with product given by wedge product and differential given by the Cartan-type formula
\begin{align*}
d_{\mathrm{CE}}\omega(\alpha_1,\dots, \alpha_{k+1}) &=
\sum_{i=1}^{k+1}(-1)^{i+1}a(\alpha_i)\omega(\alpha_1,\cdots \widehat{\alpha}_i,\dots,\alpha_{k+1})\\
&\qquad\qquad+\sum_{i<j}
(-1)^{i+j} \omega ([\alpha_i, \alpha_j], \alpha_1,\cdots \widehat{\alpha}_i,\cdots \widehat{\alpha}_j,\dots,\alpha_{k+1}).
\end{align*}
The cohomology of $(\mathrm{CE}_{\mathrm{Coh}}^\bullet(\mathbf{E}),d_{\mathrm{CE}})$ is called the Lie algebroid cohomology of $\mathbf{E}$, see \cite{Mackenzie}.
One then argues as in the previous examples, enhancing the $P_0$-algebra structure on $\mathrm{CE}^\bullet_{\mathrm{Ho}}({\mathbf{E}})$ to a $BV_0$-algebra structure.
\end{example}
\begin{example}[Xu modules]
    The existence of a volume element $\nu$ for the algebroid $\mathbf{E}$ is equivalent to the triviality of the determinant line bundle $\wedge^{\mathrm{rk}\mathbf{E}}\mathbf{E}$ on $\mathbf{M}$
and so it is quite a restrictive condition. On the other hand, such a trivialization always exists locally, and so $\mathrm{CE}_{\mathrm{Coh}}^\bullet(\mathbf{E})[\mathrm{rk} \mathbf{E}]$ is always a locally free rank 1 $\mathrm{CE}_{\mathrm{Ho}}^\bullet(\mathbf{E})$-module. In view of Remark \ref{rem:locally-free}, it is natural to wonder whether this is not sufficient to induce a $BV_0$ operator on $\mathrm{CE}_{\mathrm{Ho}}^\bullet(\mathbf{E})$. A positive answer to this question has been given by Ping Xu in \cite{Ping}. Xu's observation is that the datum of a flat $\mathbf{E}$-connection on the determinant line bundle $\wedge^{\mathrm{rk}\mathbf{E}}\mathbf{E}$ allows one to modify the definition of the $BV_0$ operator $\Delta_{\nu}$ in such a way that it becomes independent of the choice of $\nu$. These modified $BV_0$ operators can then be defined locally, thanks to the existence of local volume forms, and will define a globally defined operator since they will coincide on chart overlaps due to their independence on the choice of volume forms. This global $BV_0$ operator depends on the choice of a flat connection on $\wedge^{\mathrm{rk}\mathbf{E}}\mathbf{E}$. Remarkably, these flat connections always exist. Also, Xu's $BV_0$ operator is always an enhancement of the $P_0$-algebra structure on $\mathrm{CE}_{\mathrm{Ho}}^\bullet(\mathbf{E})$, and Xu's construction establishes a natural bijection
\[
\{\text{flat $\mathbf{E}$-connections on $\wedge^{\mathrm{rk}\mathbf{E}}\mathbf{E}$}\} \leftrightarrow 
\{\text{$BV_0$-enhancements of the $P_0$-algebra structure on $\mathrm{CE}_{\mathrm{Ho}}^\bullet(\mathbf{E})$}\}
\]
The volume form case is a particular instance of this construction: a trivialization of $\wedge^{\mathrm{rk}\mathbf{E}}\mathbf{E}$ canonically endows the determinant line bundle with a flat connection: the trivial connection.
\par
In the spirit of Section \ref{sec:cartan} it is natural to extend the constructions presented there from free rank 1 $A$-modules to \emph{Xu modules}: locally free rank 1 $A$-modules endowed with the algebraic counterpart of a flat $\mathbf{E}$-connection on the determinant bundle of a Lie algebroid. We do this in the Appendix.
\end{example}

\section{Digression: the homotopy triviality of the dgla morphism $\mathcal{L}$}
Since $\mathcal{L}_A=[\iota_a,\Bc]$, the dgla morphism $\mathcal{L}$ from Proposition \ref{prop:lie-algebra-morphism} is homotopically trivial as a morphism of cochain complexes, and $\iota\colon \mathfrak{g}_A\to \mathrm{End}_k(M)[-1]$ is a homotopy between $\mathcal{L}$ and $0$. But actually much more is true: the morphism $\mathcal{L}$ is homotopic to zero \emph{as a morphism of differential graded Lie algebras}. To see why this is true, let us recall a few basics from the theory of homotopy equivalences between morphisms of dglas. We follow the exposition in \cite{fiorenza-martinengo}

Let $\mathfrak{g}$ and $\mathfrak{h}$ be two differential graded Lie algebras. We can look at them as $L_\infty$-algebras, and so consider $L_\infty$-algebra morphisms between them. Remarkably, these are the Maurer-Cartan elements of a suitable dgla $\underline{\Hom}({\mathfrak g},{\mathfrak h})$ given by the 
Chevalley-Eilenberg complex of $\mathfrak{g}$ with coefficients in $\mathfrak{h}$ seen as a trivial $\mathfrak{g}$-module, equipped with a natural bracket $[-,-]^{}_{\underline{\Hom}}$ derived by the Lie bracket of $\mathfrak{h}$. More precisely, $\underline{\Hom}({\mathfrak g},{\mathfrak h})$ is the total dgla of the bigraded dgla
\[
\underline{\Hom}^{p,q}({\mathfrak g},{\mathfrak h})=\Hom(\wedge^q{\mathfrak g},{\mathfrak h}[p])=\Hom^p(\wedge^q{\mathfrak g},{\mathfrak h}),
\]
with $p\in \mathbb{Z}$ and $q\geq 1$, 
endowed with the Lie bracket
\[
[\,,\,]_{\underline{\Hom}}\colon \underline{\Hom}^{p_1,q_1}({\mathfrak g},{\mathfrak h})\otimes \underline{\Hom}^{p_2,q_2}({\mathfrak g},{\mathfrak h})\to \underline{\Hom}^{p_1+p_2,q_1+q_2}({\mathfrak g},{\mathfrak h})
\]
defined by
\begin{align*}
[f,g]^{}_{\underline{\Hom}}&(\gamma_1^{}\wedge\cdots\wedge\gamma_{q_1+q_2}^{})=\\
&\hskip-1em=\sum_{\sigma\in{\rm Sh}(q_1,q_2)}\pm[f(
\gamma_{\sigma(1)}\wedge\cdots \wedge\gamma_{\sigma(q_1)}),
g(
\gamma_{\sigma(q_1+1)}\wedge\cdots \wedge\gamma_{\sigma(q_1+q_2)})]^{}_{\mathfrak h},
\end{align*}
with $\sigma$ ranging in the set of $(q_1,q_2)$-unshuffles and $\pm$ standing for the Koszul sign, 
and with the differentials
\[
d_{1,0}^{}\colon
\underline{\Hom}^{p,q}({\mathfrak g},{\mathfrak h})\to \underline{\Hom}^{p+1,q}({\mathfrak g},{\mathfrak h})\]
and
\[
d_{0,1}^{}\colon
\underline{\Hom}^{p,q}({\mathfrak g},{\mathfrak h})\to \underline{\Hom}^{p,q+1}({\mathfrak g},{\mathfrak h})\]
given by
\[
(d_{1,0}^{}{f})(\gamma_1^{}\wedge\cdots\wedge\gamma_q^{})=\pm d_{\mathfrak h}(f(\gamma_1^{}\wedge\cdots\wedge\gamma_q^{}))+\sum_i
\pm f(\gamma_1\wedge\cdots
\wedge d_{\mathfrak g}\gamma_i\wedge
\cdots\wedge\gamma_{q}^{})
\]
and 
\[
(d_{0,1}^{}f)(\gamma_1\wedge\cdots\wedge\gamma_{q+1})=
\sum_{i<j}\pm f([\gamma_i,\gamma_j]^{}_{\mathfrak g}\wedge
\gamma_1\wedge\cdots
\wedge\widehat{\gamma_i}\wedge\cdots\wedge\widehat{\gamma_j}\wedge
\cdots\wedge\gamma_{q+1}^{}).
\]
An explicit determination for the signs in the above formulas can be found, e.g., in \cite{fiorenza-miti}.

One sees that the Maurer-Cartan elements of $\underline{\Hom}({\mathfrak g},{\mathfrak h})$ are those degree 1 elements $f$ in $\underline{\Hom}({\mathfrak g},{\mathfrak h})$, i.e., elements of the form
$f=\sum_{n\geq 1}f_n$ with $f_n\colon \wedge^n{\mathfrak g}\to{\mathfrak h}[1-n]$, such that
\begin{align*}
&\pm d_{\mathfrak h}f_n(\gamma_1^{}\wedge\cdots\wedge\gamma_n^{})+\frac{1}{2}\!\!\!\sum_{\substack{q_1+q_2=n\\ \sigma\in{\rm Sh}(q_1,q_2)}}\!\!\!\!\!\!\pm[f_{q_1}(
\gamma_{\sigma(1)}\wedge\cdots \wedge\gamma_{\sigma(q_1)}),
f_{q_2}(
\gamma_{\sigma(q_1+1)}\wedge\cdots \wedge\gamma_{\sigma(q_1+q_2)})]^{}_{\mathfrak h}\\
&=\sum_i
\pm f_n(\gamma_1\wedge\cdots
\wedge d_{\mathfrak g}\gamma_i\wedge
\cdots\wedge\gamma_{n}^{})\\
&\qquad\qquad+
\sum_{i<j}\pm f_{n-1}([\gamma_i,\gamma_j]^{}_{\mathfrak g}\wedge
\gamma_1\wedge\cdots
\wedge\widehat{\gamma_i}\wedge\cdots\wedge\widehat{\gamma_j}\wedge
\cdots\wedge\gamma_{n}^{}).
\end{align*}

The above is the Maurer--Cartan equation unpacked, so that one recognizes the explicit equation characterizing $L_\infty$-morphisms between $\mathfrak{g}$ and $\mathfrak{h}$ one finds  in, e.g., \cite{lada-stasheff,kontsevich}.

It is easy to exhibit at least one $L_\infty$-morphism between $\mathfrak{g}$ and $\mathfrak{h}$: the zero morphism $0$. Using the gauge action of the gauge group $\mathrm{exp}(\underline{\Hom}^0({\mathfrak g},{\mathfrak h}))$ on the set $\mathrm{MC}({\mathfrak g},{\mathfrak h})$ of Maurer-Cartan elements we obtain from $0$ the whole set of $L_\infty$-morphism from $\mathfrak{g}$ to $\mathfrak{h}$ that are gauge equivalent (or equivalently, homotopy equivalent) to the zero morphism. These are the null-homotopic $L_\infty$-morphisms between $\mathfrak{g}$ and $\mathfrak{h}$. 

It is convenient to recall the explicit form of the gauge group action on Maurer-Cartan elements:
\[
e^{-\mathcal{I}_\infty}*f=f+\sum_{n=0}^\infty \frac{(-{\rm ad}_{\mathcal{I}_\infty})^n}{(n+1)!}\ (-[\mathcal{I}_\infty,f]^{}_{\underline{\Hom}}+d^{}_{\underline{\Hom}}\mathcal{I}_\infty),
\]
where $\mathcal{I}_\infty$ is a generic degree zero element in $\underline{\Hom}({\mathfrak g},{\mathfrak h})$, i.e.,
$
\mathcal{I}_\infty=\sum_{n\geq 1}\mathcal{I}_n
$,
with $\mathcal{I}_n\colon \wedge^n\mathfrak{g}\to \mathfrak{h}[-n]$, and $d^{}_{\underline{\Hom}}= d_{0,1}+d_{1,0}$. When $f$ is the null morphism, this defines the $L_\infty$-morphism 
\[
{\mathcal{L}}_\infty=e^{-\mathcal{I}_\infty}*0=\sum_{n=0}^\infty \frac{(-{\rm ad}_{\mathcal{I}_\infty})^n}{(n+1)!}\ (d^{}_{\underline{\Hom}}\mathcal{I}_\infty),
\]
that is null-homotopic by construction, with homotopy given by $\mathcal{I}_\infty$. We can now prove the following improvement of Proposition \ref{prop:lie-algebra-morphism}.
\begin{proposition}
    Let  $(\Bc,\mathcal{L},\{-,-\})$ be a Cartan differential operator of order at most 2 and degree $1$ on $M$. Then 
    \[
    \mathcal{L}\colon (\mathfrak{g}_A,0,\{\,,\,\})\to (\mathrm{End}_k(M),[\Bc,\,],[\,,\,])
    \]
    is a null-homotopic morphism of differential graded Lie algebras, with null-homotopy given by $\iota\colon \mathfrak{g}_A\to \mathrm{End}_k(M)[-1]$.
\end{proposition}
\begin{proof}
Let us set $\mathfrak{g}=\mathfrak{g}_A$ and $\mathfrak{h}=\mathrm{End}_k(M)$. The element $\iota\colon \mathfrak{g}\to\mathfrak{h}[-1]$ can be seen as a degree zero element  $\mathcal{I}_\infty$ in $\underline{\Hom}({\mathfrak g},{\mathfrak h})$ with $\mathcal{I}_1=\iota$ and $\mathcal{I}_n=0$ for $n>1$.
 We want to show that
 \[
 e^{-\iota}*0=\mathcal{L}
 \]
 We compute
 \begin{align*}
e^{-\iota}*0&=d^{}_{\underline{\Hom}}\iota-\frac{1}{2}[\iota, d^{}_{\underline{\Hom}}\iota]_{\underline{\Hom}}+\sum_{n=2}^\infty \frac{(-{\rm ad}_{\iota})^n}{(n+1)!}\ (d^{}_{\underline{\Hom}}\iota)\\
&=d^{}_{\underline{\Hom}}\iota-\frac{1}{2}[\iota, d^{}_{\underline{\Hom}}\iota]_{\underline{\Hom}}+\sum_{n=0}^\infty \frac{(-{\rm ad}_\iota)^{n}}{(n+3)!}\ ([\iota, [\iota,d^{}_{\underline{\Hom}}\iota]_{\underline{\Hom}}]_{\underline{\Hom}}).
\end{align*}
Recalling the differentials and the brackets in our $\mathfrak{g}$ and $\mathfrak{h}$, we have
\[
d^{}_{\underline{\Hom}}\iota=d_{1,0}\iota+d_{0,1}\iota=[\iota,\Bc]+\iota_{[\,,\,]_{\mathfrak{g}}}=\mathcal{L}+\iota_{\{\,,\,\}}.
\]
Since $[\iota_a,\iota_b]=0$ for any $a,b\in \mathfrak{g}$, and $[[\mathcal{L}_a,\iota_b],\iota_c]=[\iota_{\{a,b\}},\iota_c]=0$, for any $a,b,c\in \mathfrak{g}$, we find
\[
[\iota, [\iota,d^{}_{\underline{\Hom}}\iota]_{\underline{\Hom}}]_{\underline{\Hom}}=
[\iota, [\iota,\mathcal{L}+\iota_{\{\,,\,\}}]_{\underline{\Hom}}]_{\underline{\Hom}}=0.
\]
Hence,
\[
e^{-\iota}*0=d^{}_{\underline{\Hom}}\iota-\frac{1}{2}[\iota, d^{}_{\underline{\Hom}}\iota]_{\underline{\Hom}}=
\mathcal{L}+\iota_{\{\,,\,\}}-\frac{1}{2}[\iota,\mathcal{L}+\iota_{\{\,,\,\}}]_{\underline{\Hom}}=
\mathcal{L}+\iota_{\{\,,\,\}}-\frac{1}{2}[\iota,\mathcal{L}]_{\underline{\Hom}}.
\]
We have
\[
[\iota,\mathcal{L}]_{\underline{\Hom}}(a,b)=[\mathcal{L}_a,\iota_b]-(-1)^{\deg_{\mathfrak{g}}(a)\deg_{\mathfrak{g}}(b)}[\mathcal{L}_b,\iota_a]=2\iota_{\{a,b\}}.
\]
Hence 
\[
\frac{1}{2}[\iota,\mathcal{L}]_{\underline{\Hom}}=\iota_{\{\,,\,\}}
\]
and so
\[
e^{-\iota}*0=\mathcal{L}.
\]

\end{proof}

\section{\v{S}evera bicomplexes
}
We now show how BV algebras can be naturally associated to certain spectral sequences. More precisely, we will show how, under suitably natural assumptions, the $E_2$-page together with its differential $d_2$ forms a $BV_0$-module. 

For the whole section, $(M, d_0, d_1)$ will be a bicomplex of $k$-vector spaces with differentials 
$d_0$ and $d_1$ of bidegree $(0,1)$ and $(1,0)$, respectively, and $(\pi, K)$ 
will be a vertical deformation retract for the complex $(M, d_0)$, i.e., 
$\pi \colon M \to M$ is a bidegree $(0,0)$ projection operator\footnote{i.e., $\pi^2=\pi$.} admitting a factorization
\[
\pi\colon (M,d_0) \overset{p}{\twoheadrightarrow} (H(M,d_0),0) \overset{j}{\hookrightarrow} (M,d_0),
\]
with $p$ and $j$ morphisms of cochain complexes such that $p\circ j=\mathrm{id}_{H(M,d_0)}$, and $K \colon M \to M[0,-1]$ is a bidegree $(0,-1)$ homotopy between $\mathrm{id}_M$ and $\pi$, i.e., the following equation is satisfied:
\begin{equation}\label{eq:homotopy}
\mathrm{id}_M - \pi = [d_0, K].
\end{equation}
When $(\pi,K)$ also satisfies the side conditions 
\[
\pi\circ K=0,\qquad  K\circ \pi=0, \qquad K\circ K=0.
\]
we will say that the retract $(\pi,K)$ is strong.
The homotopy operator $K$ provides $d_0$-primitives for $d_0$-exact elements of $M$: if $x=d_0y$ then
\[
d_0Kx=x-\pi x-Kd_0x=x,
\]
since $d_0x={d_0}^2y=0$ and $\pi x=jpd_0y=0$.
The following Lemma is well-known. We provide a sketch of a proof for completeness. 
\begin{lemma}\label{lem:towards-spectral}
Let $(M, d_0, d_1,\pi,K)$ as above. For any $n\geq 1$, let $\delta_n$ be the bidegree $(n,1-n)$ linear operator
    \[
    \delta_n=\underbrace{d_1\circ K\circ d_1\circ\cdots\circ d_1\circ K\circ d_1}_{\text{$n$ copies of $d_1$}}\colon M\to M.
    \]
 Set, for any $n\geq 1$,  
\[
    Z^{p,q}_n =      
    \ker(d_0\bigr\vert_{M^{p,q}})\cap \left\{ x \in M^{p,q}\, \text{ such that }\,  \delta_k x \in \mathrm{Im}(d_0) \text{ for all } 1 \le k < n \right\}.
\]
    Then we have:
    \begin{enumerate}
    \item[a)] $\delta_n(Z^{p,q}_n)\subseteq Z^{p+n,q-n+1}_n$, so that $(Z_n,\delta_n)$ is a cochain complex.
    \item[b)] elements of $Z_n^{p,q}$ are representatives in $E_0^{p,q}=M^{p,q}$ of classes that survive to $E_n^{p,q}$; taking iterated cohomology classes gives a well-defined map $\varepsilon_n\colon Z_n^{p,q}\to E_n^{p,q}$;
    \item[c)] every element in $E_n^{p,q}$ has a representative in $Z_n^{p,q}$, i.e., the map $\varepsilon_n$ is surjective;
    \item[d)] the map $\varepsilon_n$ is a morphism of cochain complexes $\varepsilon_n\colon (Z_n,\delta_n)\to (E_n,d_n)$;
    \item[e)] setting $B_n^{p,q}=\ker\{\varepsilon_n\colon Z_n^{p,q}\to E_n^{p,q}\}$ we have an isomorphism of cochain complexes
    \[
    (Z_n/B_n,\overline{\delta}_n)\xrightarrow{\sim} (E_n,d_n),
    \]
    where $\overline{\delta}_n$ is the morphism induced by $\delta_n$ on the quotient.    
    \end{enumerate}
\end{lemma}

\begin{proof}
For $n=1$ we have $\delta_1=d_1$ and $Z_1^{p,q}=\ker(d_0\bigr\vert_{M^{p,q}})$. Since $[d_0,d_1]=0$, we have that $d_1$ maps  $\ker(d_0\bigr\vert_{M^{p,q}})$ to $\ker(d_0\bigr\vert_{M^{p+1,q}})$ and so to $Z_1^{p+1,q}$. Assume now $n\geq 2$. From
$[d_0,d_1]=0$, $[d_0,K]=\mathrm{id}_M-\pi$, and $d_1^2=0$, we obtain
\begin{equation}\label{eq:recursion}
[d_0, \delta_{n}] = \sum_{k=1}^{n-1} (-1)^{k-1} \delta_{n-k}  \pi \delta_k.
\end{equation}
If $x\in Z^{p,q}_{n}$, then $\delta_k x \in \mathrm{Im}(d_0)$ for $1\leq k\leq n-1$. Since $\pi$ vanishes on the image of $d_0$ we find $[d_0,\delta_{n}](x)=0$ for any $x\in Z^{p,q}_{n}$. On the other hand, $Z^{p,q}_{n}\subseteq \ker(d_0\bigr\vert_{M^{p,q}})$, so $[d_0,\delta_{n}](x)=d_0\delta_n x$ for any $x\in Z^{p,q}_{n}$. This gives $d_0\delta_n x=0$, i.e., $\delta_n$ maps $Z^{p,q}_{n}$ to $\ker(d_0\bigr\vert_{M^{p+n,q-n+1}})$. From the definition of $\delta_n$, one sees that $\delta_k\delta_n=0$ for any $k,n\geq 1$, since two consecutive copies of $d_1$ occur in the composition. Hence, if $y=\delta_nx$, then $\delta_ky=0\in\mathrm{Im}(d_0)$, for all $1\leq k<n$.
Thus $\delta_n x\in Z_n^{p+n,q-n+1}$, i.e., $\delta_n$ maps $Z^{p,q}_{n}$ to $Z_n^{p+n,q-n+1}$. Clearly ${\delta_n}^2=0$, so this proves statement {\it a)}. Statements {\it b)}, {\it c)} and {\it d)} are the classical zig-zag description of the $n$-th page $E_n$ of the spectral sequence associated with a bicomplex, see, e.g. \cite[Chapter 3]{bott-tu}. More precisely, the homotopy $K$ provides canonical $d_0$-primitives at each stage of the zig-zag procedure, so that the resulting
higher differentials are represented by the operators $\delta_n$.\footnote{Or by the operators $(-1)^n\delta_n$, depending on conventions.} Finally, thanks to statements {\it c)} and {\it d)}, statement {\it e)} is simply the first isomorphism theorem for complexes applied to the morphism $\varepsilon_n$. 
\end{proof}

By a slight abuse of notation, if $A$ is a graded commutative algebra and $N$ a bigraded vector space, we say that $\rho\colon A\otimes N\to N$ is a graded $A$-module structure on $N$ to mean it is a graded $A$-module structure on the total graded vector space $\mathrm{tot}(N)$, where $\mathrm{tot}(N)^i=\oplus_{p+q=i}N^{p,q}$.
\begin{definition}\label{def:severa-bicomplex}
A \emph{\v{S}evera bicomplex} is a 7-ple $(M,d_0,d_1,\pi,K,A,\rho)$ where
\begin{itemize}
    \item $(M,d_0,d_1)$ is a bicomplex of $k$-vector spaces with differentials 
$d_0$ and $d_1$ of bidegree $(0,1)$ and $(1,0)$;
    \item $(\pi,K)$ is a vertical strong deformation retract for the complex $(M, d_0)$;
    \item $A$ is a graded commutative algebra;
    \item $\rho\colon A\otimes M\to M$ is a graded $A$-module structure on $M$;
    \item $d_0\in \Diff_A^{\leq 0}(M,M)$;
    \item $d_1\in \Diff_A^{\leq 1}(M,M)$;
    \item $K\in \Diff_A^{\leq 0}(M,M)$.
\end{itemize}
We say that a \v{S}evera bicomplex has level $\geq N$ if the linear subspaces $Z_n$ and $B_n$ defined in Lemma \ref{lem:towards-spectral} are sub-$A$-modules of $M$ for any $1\leq n\leq N$.
\end{definition}
\begin{remark}
    Notice that the conditions on the order of the differentials $d_0$ and $d_1$ can be compactly expressed as $d_i\in \Diff_A^{\leq i}(M,M)$, for $i=0,1$.
\end{remark} 
\begin{remark}\label{rem:at-least-one}
 Each \v{S}evera bicomplex has level at least 1. Indeed $Z_1=\ker(d_0)$ and $d_0$ is an $A$-linear operator since $\Diff^{\leq0}_A(M,M)=\mathrm{End}_A(M)$. Since $E_1=H(M,d_0)$, the map $\varepsilon_1$ 
 is the canonical projection $\ker(d_0)\to H(M,d_0)$ and so its kernel $B_1$ is $\mathrm{Im}(d_0)$, which is again an $A$-submodule of $M$.
\end{remark}
\begin{remark}
    From Proposition \ref{prop:multiplicative-filtration} it immediately follows from the definition that the projection operator $\pi$ in a \v{S}evera bicomplex is an element in $\Diff_A^{\leq 0}(M,M)$, i.e., it is a morphism of $A$-modules from $M$ to $M$.
\end{remark}
The crucial property of \v{S}evera bicomplexes is the following.
\begin{proposition}\label{prop:in-diff-n}
    Let $(M,d_0,d_1,\pi,K,A,\rho)$ be a \v{S}evera bicomplex of level $\geq n$, and let 
    \[
    d_n\colon E_n\to E_n
    \]
    the differential of the $n$-th page of the spectral sequence associated with the bicomplex $(M,d_0,d_1)$. Then $d_n\in \Diff_A^{\leq n}(E_n,E_n)$, where the $A$-module structure on $E_n$ is the one induced by the $A$-submodule structures on $Z_n$ and $B_n$.
\end{proposition}
\begin{proof}
    For $n=0$, we have $E_0=M$ and $d_0\in \Diff^{\leq 0}(E_0,E_0)$ by assumption. For $n\geq 1$, Lemma \ref{lem:towards-spectral} tells us that the differential $d_n\colon E_n\to E_n$ is induced by a subquotient construction by the linear operator
    \[
    \delta_n=\underbrace{d_1\circ K\circ d_1\circ\cdots\circ d_1\circ K\circ d_1}_{\text{$n$ copies of $d_1$}}\colon M\to M.
    \]
    By Proposition \ref{prop:multiplicative-filtration} we have $\delta_n\in \Diff_A^{\leq n}(M,M)$ and so, by Proposition \ref{prop:subquotients} we have $d_n\in \Diff_A^{\leq n}(E_n,E_n)$. 
\end{proof}
\begin{corollary}
    Let $(M,d_0,d_1,\pi,K,A,\rho)$ be a \v{S}evera bicomplex of level $\geq 2$. If $E_2$ is a free rank 1 graded $A$-module, then $(E_2,d_2)$ is a $P_0$-module for $A$. { In particular, $d_2$ is a $BV_0$-operator on $E_2$.}
\end{corollary}
\begin{proof}
 We know from Proposition \ref{prop:in-diff-n} that $d_2\in \mathrm{Diff}^{\leq 2}(E_2,E_2)$. Moreover $d_2$ has total degree 1 and ${d_2}^2=0$. 
\end{proof}
\begin{corollary}\label{cor:P0-from-E2}
    Let $(M,d_0,d_1,\pi,K,A,\rho)$ be a \v{S}evera bicomplex of level $\geq 2$ such that $E_2$ is a free rank 1 graded $A$-module. Then $A$ carries a canonical $P_0$-algebra structure induced by the \v{S}evera bicomplex. In particular, $A$ with the opposite grading is canonically a Gerstenhaber algebra. 
\end{corollary}

\begin{corollary}
    Let $(M,d_0,d_1,\pi,K,A,\rho)$ be a \v{S}evera bicomplex such that $E_2$ is a free rank 1 graded $A$-module. Then each $d_2$-closed basis element $\zeta$ for the free  rank 1 $A$-module $E_2$ induces a $BV_0$ algebra on $A$, enhancing the $P_0$-algebra structure from Corollary  \ref{cor:P0-from-E2}. Taking the opposite grading on $A$ one obtains a Batalin-Vilkovisky algebra structure on $A$ enhancing the Gerstenhaber algebra structure from Corollary  \ref{cor:P0-from-E2}.
\end{corollary}
A useful sufficient condition for a \v{S}evera bicomplex to have level $\geq 2$ is provided by the following Lemma.
\begin{lemma}\label{lemma:zero-d_1-implies-level_geq2}
    Let $(M,d_0,d_1,\pi,K,A,\rho)$ be a \v{S}evera bicomplex. If $d_1\colon E_1\to E_1$ is the zero morphism, then $(M,d_0,d_1,\pi,K,A,\rho)$ has level $\geq 2$. 
\end{lemma}
\begin{proof}
    The statement is informally clear: we know from Remark \ref{rem:at-least-one} that $E_1$ inherits from $E_0=M$ an $A$-module structure. If $d_1\colon E_1\to E_1$ is the zero morphism then $E_2=E_1$ and so also $E_2$ inherits an $A$-module structure, and this is essentially the property of being of level $\geq 2$. A formal proof goes as follows. Assume $d_1\colon E_1\to E_1$ is the zero morphism. Then,
    $d_1\colon M\to M$ maps $d_0$-closed elements to $d_0$-exact elements, i.e., $d_1(\ker(d_0))\subseteq \mathrm{Im}(d_0)$. Since $\delta_1=d_1$, this implies that $Z_2=\ker(d_0)$, and this is an $A$-submodule of $M$ thanks to the $A$-linearity of $d_0$. Now we identify the linear subspace $B_2$. Generally speaking, an element $\xi$ in $Z_2$ induces the zero class in $E_2$ if and only if $\xi=d_1\eta+d_0\beta$ for some $\eta\in \ker(d_0)$. Therefore, generally $B_2=d_1(\ker(d_0))+\mathrm{Im}(d_0)$. Under the assumption that $d_1\colon E_1\to E_1$ is the zero morphism, we have $d_1(\ker(d_0))\subseteq \mathrm{Im}(d_0)$, and so $B_2=\mathrm{Im}(d_0)$. Again, this is an $A$-submodule of $M$ thanks to the $A$-linearity of $d_0$. Notice that, in particular, we recovered that, under the assumption that $d_1\colon E_1\to E_1$ is the zero morphism, we have $E_2=\ker(d_0)/\mathrm{Im}(d_0)=E_1$.
    \end{proof}
 \begin{remark}\label{rem:palladio}
  When $d_1\colon E_1\to E_1$ is the zero morphism, the condition $d_2\zeta=0$ and the fact that $\zeta$ is a basis element for the free  rank 1 $A$-module $E_2$ can be expressed in terms of a lift $\xi$ of $\zeta$ to $M$ by requiring: 
  \begin{itemize}
  \item $d_0\xi=0$;
  \item $d_1\xi=d_0\eta$, for some $\eta\in M$;\footnote{When $d_1\colon E_1\to E_1$ is the zero morphism, this condition is automatically implied by the condition $d_0\xi=0$. Here we are explicitly including it  to exactly reproduce the definition of Palladio element from \cite{fiorenza-kowalzig2018}.}
  \item the map of cochain complexes
  \[
  \iota_{(-)}\xi\colon (\mathfrak{g}_A,0)\to (M[-1],d_0)
  \]
  is a quasi-isomorphism.
  \end{itemize}
  Elements $\xi$ of this type appear under the name of Palladio elements in \cite{fiorenza-kowalzig2018}.
 \end{remark}

 \begin{example}\label{ex:twisted-de-Rham}
     A simple example of a \v{S}evera bicomplex is the following. Let $\mathbf{M}$ be a smooth manifold admitting a never zero closed $1$-form $\theta$. For instance, $\mathbf{M}$ could be a product manifold $\mathbf{M}=S^1\times \mathbf{N}$, and $\theta$ be the pullback of the standard angular one-form on $S^1$. The $\theta$-twisted de Rham differential on the $C^\infty(\mathbf{M})$-module $\mathcal{E}^\bullet(\mathbf{M})$ of differential forms on $\mathbf{M}$ is given by
     \[
     d_\theta \eta= d\eta+\theta\wedge \eta.
     \]
 Let us write $d_0=\theta\wedge-$ and $d_1=d$. Finally, let $v$ be a vector field on $\mathbf{M}$ such that $\theta(v)\equiv 1$. Since $\theta$ is always nonzero, such a vector field always exists. For instance, one can endow $\mathbf{M}$ with a Riemannian metric, and consider the vector field $w_\theta$ defined by the identity $\theta=\langle w_\theta,-\rangle$. Since $\theta$ is nowhere vanishing, also $w_\theta$ is never zero. Then one can take $v=w_\theta/\|w_\theta\|^2$. Let $K=\iota_v$ be the contraction with $v$. Then we have $[K,\theta\wedge-]=\mathrm{id}$, and so $(\mathcal{E}^\bullet(\mathbf{M}),\theta\wedge-,d,0,\iota_v,C^\infty(\mathbf{M}),\cdot)$ is a \v{S}evera bicomplex.
 \end{example}

\begin{remark}
Since $[K,\theta\wedge-]=\mathrm{id}$, the complex $(\mathcal{E}^\bullet(\mathbf{M}),\theta_\wedge-)$ of Example \ref{ex:twisted-de-Rham} is acyclic, and so  $E_1 = H(\mathcal{E}^\bullet(\mathbf{M}), \theta \wedge)=0$. In particular, $(\mathcal{E}^\bullet(\mathbf{M}),\theta\wedge-,d,0,\iota_v,C^\infty(\mathbf{M}),\cdot)$ has trivially level $\geq 2$. Yet, $E_2=0$ and so $E_2$ is not a free rank 1 $C^\infty(\mathbf{M})$-module. In the next Section we are going to see how a variant of Example \ref{ex:twisted-de-Rham} will produce a a free rank 1 $C^\infty(\mathbf{M})$-module: the manifold $\mathbf{M}$ will be an odd symplectic manifold, and we will be considering the twisted de Rham operator with twist given by the odd symplectic form.
\end{remark}

\section{The \v{S}evera bicomplex of an odd symplectic manifold}

Differential forms on a graded manifold $\mathbf{M}$ are naturally bigraded by form degree $f$ and internal degree\footnote{The degree coming from the degree of coordinates.} $i$. By an odd symplectic manifold of dimension $(n|n)$, we mean a pair $(\mathbf{M},\omega)$, where  $\mathbf{M}$ is a graded manifold of even dimension $n$ and odd dimension $n$, and $\omega$ is a symplectic form of bidegree $(2,-1)$. Locally, $\mathbf{M}$ carries Darboux coordinates $q^i, p_i$ of degree $0$ and $-1$ respectively. In these coordinates
\[
\omega=\sum_{i=1}^ndq^i\wedge dp_i.
\]
The multiplication operator by the symplectic form $\omega$ has bidegree $(2,-1)$, while the de Rham differential $d$ has bidegree $(1,0)$. We make an affine change in the bigrading in order to have $\omega\wedge-$ of bidegree $(0,1)$ and $d$ of bidegree $(1,0)$: we set
\[
p=f+2i+n; \qquad q=-i
\]
and 
\[
M=\bigoplus_{p,q} M^{p, q}
\]
with $M^{p,q}=\Omega^{-p+2q-n,-q}$, where $\Omega^{f,i}$ is the space of differential forms of bidegree $(f,i)$. 

\begin{lemma}\label{M_is_bicomplex}
Let $d_0\colon M\to M[0,1]$ and $d_1\colon M\to M[1,0]$ the linear operators $d_0=\omega\wedge -$ and $d_1=d$, where $d$ denotes the de Rham differential. Then $(M,d_0,d_1)$ is a bicomplex.
\end{lemma}
\begin{proof}
The condition ${d_0}^2=0$ follows from the fact $\omega$ has total degree 1 and the wedge product on $\Omega$ is graded commutative. The condition ${d_1}^2=0$ is the usual property of the de Rham operator being a differential. Finally, the condition $[d_0,d_1]=0$ follows from the Leibniz rule for $d$, since $\omega$ is closed of total odd degree.
\end{proof}

\begin{remark}\label{rem:A-module}
Since the change in the bigrading $(f,i)\leftrightarrow (p,q)$ is an affine one and not a linear one, the wedge product of differential forms does not induce a bigraded algebra structure on $M$. Yet, it induces a graded $\Omega$-module structure with respect to the total gradings $f+i$ and $p+q$:
\[
\Omega^{f,i}\otimes M^{p,q}\to M^{f+p+2i,q-i}.
\]
Let $A=C^\infty(\mathbf{M})$ be the graded commutative algebra of smooth functions on $\mathbf{M}$. We have $A^i=\Omega^{0,i}$, so that the inclusion of graded commutative algebras $A\hookrightarrow \Omega$ naturally makes $M$ a graded $A$-module.
\end{remark}
 \begin{notation}
We denote by $\rho$ the natural graded $A$-module structure on $M$ mentioned in Remark \ref{rem:A-module}. 
\end{notation}
\begin{lemma}\label{lem:leibniz}
The differentials $d_0$ and $d_1$ are differential operators of order $0$ and $1$, respectively, with respect to the natural action of $A=C^\infty(\mathbf{M})$ on $M$.    
\end{lemma}
\begin{proof}
The $A$-linearity of $d_0$ is manifest. Concerning $d_1$, the fact that $d_1\in \Diff^{\leq 1}_A(M,M)$ is the Leibniz rule for the de Rham differential: $d_1(a\eta)=d_1a\wedge \eta+(-1)^{\mathrm{deg}_A(a)}a\, d_1\eta$, hence $[d_1,\rho_a]$ is the $A$-linear operator $d_1a\wedge-$.    
\end{proof}

Let $\Pi$ be the bivector field on $\mathbf{M}$ dual to the symplectic form $\omega$, i.e., the bivector field defining the odd bracket $\{-,-\}$ on functions on $\mathbf{M}$. Let $\iota_\Pi\colon \Omega\to \Omega[-2,1]$ be the contraction operator with $\Pi$. We can look at $\iota_\Pi$ as  linear operator
\[
\iota_\Pi\colon M\to M[0,-1].
\]
\begin{lemma}\label{lem:comm_omega_pi}
The linear operator 
\[
\Lambda=[d_0,\iota_\Pi]\colon M\to M
\]
is semisimple with integer eigenvalues. Moreover $\Lambda\in \Diff^{\leq 0}(M,M)$.
\end{lemma}
\begin{proof}
    Since both $d_0$ and $\iota_\pi$ are endomorphisms of the bundle $\bigoplus^k\wedge^kT^*\mathbf{M}$ we can work pointwise with a linear basis $dq^i, dp_i$ of the cotangent space at some point $x$ such that $\omega$ is in standard Darboux form in this basis. On a basis element $dq^I\wedge dp_J$ of  
    $\bigoplus^k\wedge^kT^*_x\mathbf{M}$, using $\iota_\Pi=\iota_{\partial_{q^i}}\iota_{\partial_{p_i}}$ 
    and the fact that $\iota_{\partial{q^i}}$ and $\iota_{\partial{p_i}}$ are graded derivations of $(\Omega,\wedge)$ we find
\[
\Lambda\colon dq^I\wedge dp_J \mapsto (n-|I|+|J|) dq^I\wedge dp_J
\]
and so the restriction of $\Lambda$ to the exterior algebra of the cotangent bundle is semisimple with integer eigenvalues. This splits at each point the exterior algebra of the cotangent bundle into a direct sum of $\Lambda$-eigenspaces, and so splits the whole of $\Omega$, and so the whole of $M$, into a direct sum of $\Lambda$-eigenspaces. Finally, $\Lambda$ is $A$-linear since both $d_0$ and $\iota_\Pi$ are $A$-linear.
\end{proof}
\begin{notation}
 Let us write 
\[
M=\bigoplus_{k\in \mathbb{Z}}V_k
\]
for the decomposition of $M$ into $\Lambda$-eigenspaces. Notice that, since $|I|\leq n$, we have $(n-|I|+|J|)\geq 0$ so that the decomposition is actually
\[
M=\bigoplus_{k\geq 0}V_k.
\]
\end{notation}
\begin{remark}\label{rem:zero-on-V0}
On $V_0$ we have $|I|=n+|J|$. Since $|I|\leq n$, this forces $|I|=n$ and $|J|=0$. In other words, elements of $V_0$ can be written locally in a Darboux chart $U$ as
\begin{equation}\label{eq:local}
\eta\bigr\vert_U= f(\vec{p},\vec{q}) dq^1\wedge\cdots\wedge dq^n
\end{equation}
for some smooth function $f$.
\end{remark}
\begin{lemma}\label{M_is_Ssubcomplex_and_Amodule}
The decomposition of $M$ as the direct sum of the $\Lambda$-eigenspaces is a decomposition of $(M,d_0)$ into a direct sum of subcomplexes:
\[
(M,d_0) =\bigoplus_{k\geq 0} (V_k,d_0).
\]
Moreover, each of the $V_k$ is an $A$-module.
\end{lemma}
\begin{proof}
The operator $d_0$ commutes with $\Lambda$, since
\[
[d_0,\Lambda]=[d_0,[d_0,\iota_\Pi]]=[[d_0,d_0],\iota_\Pi]
-[d_0,[d_0,\iota_\Pi]]=
-[d_0,[d_0,\iota_\Pi]]=-[d_0,\Lambda].
\]
This proves that the $V_k$'s are subcomplexes of $(M,d_0)$. Since $\Lambda$ is $A$-linear, its eigenspaces are $A$-submodules of $M$. This proves the final part of the statement.
\end{proof}
\begin{lemma}\label{lem:Pi-preserves}
 The linear operator $\iota_\Pi$ preserves the $\Lambda$-eigenspaces $V_k$.   
\end{lemma}
\begin{proof}
The operator $\iota_\Pi$ commutes with $\Lambda$, since
\[
[\iota_\Pi,\Lambda]=[\iota_\Pi,[d_0,\iota_\Pi]]=
[[\iota_\Pi,d_0],\iota_\Pi]-[d_0,[\iota_\Pi,\iota_\Pi]]=[[d_0,\iota_\Pi],\iota_\Pi]=[\Lambda,\iota_\Pi]=-[\iota_\Pi,\Lambda].
\]
\end{proof}
\begin{notation}\label{not:-pi-K}
We denote by $K\colon M\to M[0,-1]$ the linear operator defined by
 \[
    K=\bigoplus_{k>0}\frac{1}{k}\iota_\Pi\bigr\vert_{V_k},
\]
and by $\pi\colon M\to M$ the linear operator given by projection on the $\Lambda$-eigenspace $V_0$ followed by the inclusion $V_0\hookrightarrow M$. 
\end{notation}
\begin{lemma}
    The pair $(\pi,K)$ from Notation \ref{not:-pi-K} is a vertical strong deformation retract for the complex $(M, d_0)$. Moreover, $K\in \Diff^{\leq 0}_A(M,M)$.
\end{lemma}
\begin{proof}
By Remark \ref{rem:zero-on-V0}, the differential $d_0$ vanishes on $V_0$.
  By Lemmas \ref{M_is_Ssubcomplex_and_Amodule} and \ref{lem:Pi-preserves}, both $d_0$ and $\iota_\Pi$ preserve the $\Lambda$-eigenspaces $V_k$. Hence,
\[
[d_0,K]=\bigoplus_{k>0}[d_0,\frac{1}{k}\ \iota_\Pi]\bigr\vert_{V_k}=\bigoplus_{k>0}\mathrm{Id}_{V_k}=\mathrm{Id}_M-\pi.
\]
Therefore $(V_0,0)\hookrightarrow (M,d_0)$ is a quasi-isomorphism. Since $K$ preserves the $\Lambda$-eigenspaces $V_k$ and vanishes on $V_0$, the pair $(\pi,K)$ trivially satisfies the side conditions $K\circ \pi=0$ and $\pi\circ K=0$. To check that $K\circ K=0$ we only need to check this on each eigenspace $V_k$. If $k=0$ the identity $K\circ K\vert_{V_0}=0$ is trivially satisfied. For $k\neq 0$ we find
\[
K\circ K\bigr\vert_{V_k} =  \frac 1 {k^2} \iota_\Pi \circ  \iota_\Pi \bigr\vert_{V_k} = \ \frac 1 {k^2} \iota^{}_{\Pi\wedge \Pi} \bigr \vert_{V_k} =0,
\]
since $\Pi$ is odd, and so it satisfies $\Pi\wedge \Pi=0$. Finally, since the decomposition of $M$ into the direct sum of the $\Lambda$-eigenspaces $V_k$ is a decomposition in sub-$A$-modules, to check the $A$-linearity of $K$ we only need to check that the restriction of $K$ to each $V_k$ is $A$-linear. This is clear, since $\iota_\Pi$ is $A$-linear.
\end{proof}
\begin{corollary}\label{cor:acyclic}
All of the subcomplexes $(V_k,d_0)$, with the exception of $V_0$ are acyclic.
\end{corollary}
\begin{lemma}\label{lem:even-columns}
We have 
\[
M^{\mathrm{even},\bullet}=\bigoplus_{k \, \mathrm{even}}V_k; \qquad M^{\mathrm{odd},\bullet}=\bigoplus_{k \, \mathrm{odd}}V_k;
\]
\end{lemma}
\begin{proof}
 In the notation used in the proof of Lemma \ref{lem:comm_omega_pi}, the form degree $f$ of an element in $V_k$ is $|I|+|J|$. Since $k=n-|I|+|J|$, we find $f+n\equiv k \mod 2$. On the other hand, the $(p,q)$-grading on $M$ is related to the $(f,i)$-grading on $\Omega$ by $p=f+2i+n$ and $q=-i$, so that $p\equiv f+n\mod 2$. Therefore, $p\equiv k\mod 2$, and the conclusion follows.
\end{proof}
\begin{corollary}\label{cor:E1=E2}
 The odd columns of $M$ are $d_0$-acyclic. As a consequence the odd columns of $E_1=H(M,d_0)$ are zero, hence the differential $d_1\colon E_1\to E_1$ is zero.   
\end{corollary}
\begin{proof}
 Immediate from Lemma \ref{lem:even-columns}  and Corollary \ref{cor:acyclic}.  
\end{proof}
 Recalling Lemma \ref{lemma:zero-d_1-implies-level_geq2}, all of the above can be summarized in the following statement, which is the main result of this Section.
 \begin{theorem}
   Let $(\mathbf{M},\omega)$ be an odd symplectic manifold of dimension $(n|n)$. Assume the notation from above. Then $(M,d_0,d_1,\pi,K,A,\rho)$ is a \v{S}evera bicomplex of level $\geq 2$.
 \end{theorem}
A striking property of the \v{S}evera bicomplex of differential forms on the odd symplectic manifold $(\mathbf{M},\omega)$ is the following.
 
\begin{theorem}[\v{S}evera] Let $(\mathbf{M},\omega)$ and $(M,d_0,d_1,\pi,K,A,\rho)$ as above. Then $E_2$ is canonically isomorphic to the $A$-module $\Dens^{1/2}(\mathbf{M})$ of half-densities on the graded manifold $\mathbf{M}$. In particular $E_2$ is a free rank 1 $A$-module. 
\end{theorem}
\begin{proof}
Since $E_2=E_1$ one is reduced to showing that $E_1=H(M,d_0)$ is naturally isomorphic to the $C^\infty(\mathbf{M})$-module $\Dens^{1/2}(\mathbf{M})$ of half-densities on the graded manifold $\mathbf{M}$. This has been shown by \v{S}evera in \cite{SeveraNonommutativeDifferentialForms}. For completeness, here we reproduce a sketch of the proof, following the exposition in \cite{Severa_2006}. A direct way of proving the statement consists in identifying $H(M,d_0)$ with $V_0$ by the quasi-isomorphism $(V_0,0)\hookrightarrow (M,d_0)$, and noticing that in every Darboux chart an element in $V_0$ has a unique local expression of the form \eqref{eq:local}. Then one checks that under a change of Darboux charts the coefficient $f\in C^\infty(\mathbf{M})$ gets multiplied by the square root of the Berezinian of the change of coordinates. A more conceptual proof uses Manin's  cohomological definition of the Berezinian \cite{manin}, directly linking half-densities to the cohomology of $(M,d_0)$. Given a finite dimensional super vector space $V$, Manin considers the super vector space $W=V\oplus V^*[1]$, endowed with the canonical degree $-1$ odd nondegenerate bilinear form $\omega\in \bigwedge^2W^*$ given by the pairing between $V$ and $V^*$. This induces the differential $d_0=\omega\wedge-$ on the exterior algebra $\bigwedge^\bullet W^*$ so Manin defines $\mathrm{Ber}(V)$ as $\mathrm{Ber}(V)=H(\bigwedge^\bullet W^*, d_0)$. Starting with $V^*[1]$, one has $(V^*[1])^*[1]=V$, so that $V^*[1]\oplus (V^*[1])^*[1]$ is canonically isomorphic to $W$ and one finds a canonical isomorphism between $\mathrm{Ber}(V)$ and $\mathrm{Ber}(V^*[1])$. Manin then shows that $\mathrm{Ber}$ is an exponential functor, i.e., that given two super vector spaces there is a canonical isomorphism $\mathrm{Ber}(V_1\oplus V_2)\cong \mathrm{Ber}(V_1)\otimes \mathrm{Ber}(V_2)$. This gives a canonical isomorphism
\[
\mathrm{Ber}(W)=\mathrm{Ber}(V\oplus V^*[1])\cong \mathrm{Ber}(V)\otimes \mathrm{Ber}(V^*[1])\cong \mathrm{Ber}(V)^{\otimes 2}=H(\textstyle{\bigwedge^\bullet W^*}, d_0)^{\otimes 2}.
\]
Since every odd symplectic vector space $W$ can be split as $W=V\oplus V^*[1]$ by choosing $V$ to be a Lagrangian subspace of $(W,\omega)$, this associates with any Lagrangian splitting of $W$ a distinguished isomorphism $H(\bigwedge^\bullet W^*, d_0)^{\otimes 2}\cong \mathrm{Ber}(W)$. One then shows that this isomorphism is independent of the Lagrangian splitting, thus giving a canonical isomorphism $H(\bigwedge^\bullet W^*, d_0)\cong \mathrm{Ber}^{1/2}(W)$. Taking $W=T_x\mathbf{M}$ at each point of the odd symplectic supermanifold $\mathbf{M}$ concludes \v{S}evera's construction of the canonical isomorphism $E_2\cong \Dens^{1/2}(\mathbf{M})$.
 \end{proof}

\begin{corollary}
$C^\infty(\mathbf{M})$ carries a canonical $P_0$-algebra structure induced by the \v{S}evera bicomplex. In particular, $C^\infty(\mathbf{M})$ with the opposite grading is canonically a Gerstenhaber algebra.
\end{corollary}

{
\begin{corollary}
The space  $\Dens^{1/2}(\mathbf{M})$ of half-densities on an odd symplectic manifold $(\mathbf{M},\omega)$ is naturally a $P_0$-module for the graded commutative algebra $C^\infty(\mathbf{M})$. In particular, via the \v{S}evera's canonical isomorphism $E_2\cong \Dens^{1/2}(\mathbf{M})$, the differential $d_2$ defines a canonical $BV_0$-operator $\Delta_{\mathrm{Dens}^{1/2}}$ on the space of half-densities. Taking the opposite grading, one gets a canonical Batalin-Vilkovisky operator on the space of half-densities.
\end{corollary}
}
In order to prove the next result, it will be useful the following technical lemma.
\begin{lemma}\label{lem:p0-suffices}
In the notation above, let $\varphi\in \mathrm{End}_k(E_2)$ and let $\widetilde{\varphi}\colon M\to M$ be a $k$-linear operator inducing $\varphi$, that is, $\widetilde{\varphi}(Z_2)\subseteq Z_2$, $\widetilde{\varphi}(B_2)\subseteq B_2$ and the induced operator on the quotient $E_2=Z_2/B_2$ is $\varphi$. If $\pi\circ\widetilde{\varphi}\circ \pi=0$, then $\varphi=0$.  Equivalently, if $j\colon V_0\hookrightarrow M$ denotes the inclusion, we have that $\pi\circ \widetilde{\varphi}\circ j=0$ implies $\varphi=0$.   
\end{lemma}
\begin{proof}
    Since by Corollary \ref{cor:E1=E2} we have that $d_1:E_1\to E_1$ is the zero morphism, we know from the proof of Lemma \ref{lemma:zero-d_1-implies-level_geq2} that $Z_2=\ker(d_0)$ and $B_2=\mathrm{Im}(d_0)$. 
    From 
    \[
    \pi=\mathrm{Id}_M-[d_0,K]
    \]
    we obtain
    \[
    \pi\circ \widetilde{\varphi}\circ \pi=\widetilde{\varphi}-\widetilde{\varphi}\circ[d_0,K]-[d_0,K]\circ \widetilde{\varphi}+[d_0,K]\circ\widetilde{\varphi}\circ [d_0,K].
    \]
    Therefore, if $\pi\circ\widetilde{\varphi}=0$ then 
    \[
    \widetilde{\varphi}=\widetilde{\varphi}\circ[d_0,K]+[d_0,K]\circ \widetilde{\varphi}-[d_0,K]\circ\widetilde{\varphi}\circ [d_0,K].
    \]
    We have $[d_0,K]\colon \ker(d_0)\to \mathrm{Im}(d_0)$. Hence, the fact that $\widetilde{\varphi}$ preserve $Z_2$ and $B_2$ implies that $\widetilde{\varphi}(Z_2)\subseteq B_2$ and so $\varphi=0$. To conclude, we have to show that $\pi\circ \widetilde{\varphi}\circ j=0$ if and only if $\pi\circ \widetilde{\varphi}\circ \pi=0$. This is immediate from the fact that $\pi$ is the projection on $V_0$.
\end{proof}

\begin{proposition}\label{prop:p0-suffices}
 The $P_0$- (or Gersthenhaber) bracket induced by the \v{S}evera bicomplex on $C^\infty(\mathbf{M})$ is the canonical degree $1$ (or $-1$) Poisson bracket $\{\,,\,\}_{\Pi}$ defined by the Poisson bivector $\Pi$, i.e., the bracket
 \[
 \{f,g\}_\Pi=\iota_\Pi(df\wedge dg).
 \]
\end{proposition}
\begin{proof}
For any $\eta\in \Omega$, let us denote by $\rho_\eta\colon M\to M[\deg_\Omega \eta]$ the $A$-linear operator $\rho_\eta=\eta\wedge-$. With this notation, the Leibniz identity for the de Rham differential reads $[d_1,\rho_\eta]=\rho_{d_1 \eta}$, for any $\eta\in \Omega$. 
The $P_0$-bracket on $C^\infty(\mathbf{M})$ induced by the \v{S}evera bicomplex is characterized by the equation
\[
[[d_2, \iota_f],\iota_g]=\iota_{\{f,g\}},\qquad \text{for any}\qquad f,g\in \mathfrak{g}_A
\]
i.e.,
\[
[[d_1Kd_1, \rho_f],\rho_g]=\rho_{\{f,g\}},\qquad \text{for any}\qquad f,g\in A,
\]
in $\mathrm{End}_k(E_2)$. Since $E_2$ is a rank 1 free $A$-module, showing that $\{\,,\,\}=\{\,,\,\}_\Pi$ is therefore equivalent to showing that the $k$-linear endomorphism $\varphi_{f,g}$ of $M$ given by $\varphi_{f,g}=[[d_1Kd_1, \rho_f],\rho_g]-\rho_{{\{f,g\}}_\Pi}$ induces the zero morphism on $E_2$. In order to apply Lemma \ref{lem:p0-suffices}, we show that $\varphi_{f,g}(Z_2)\subseteq Z_2$ and $\varphi_{f,g}(B_2)\subseteq B_2$. By $A$-linearity, both $\rho_f,\rho_g$ and $\rho_{\{f,g\}}$ preserve $Z_2=\ker(d_0)$ and $B_2=\mathrm{Im}(d_0)$. The fact that $\delta_2=d_1Kd_1$ preserves $Z_2$ and $B_2$ has been shown in 
Lemma \ref{lem:towards-spectral}. Hence $\varphi_{\{f,g\}}$ preserves both $Z_2$ and $B_2$ and by Lemma \ref{lem:p0-suffices} in order to prove that $\varphi_{\{f,g\}}$ induces the zero morphism on $E_2$ we only need to show that $\pi\circ \varphi_{\{f,g\}}\circ j=0$.

Since $K$ is $A$-linear, we have $[K,\rho_f]=[K,\rho_g]=0$ for every $f,g\in \mathfrak{g}_A$. Using this, the identity $[\rho_\eta,\rho_\beta]=0$ for any $\eta,\beta\in \Omega$, and the Leibniz identity $[d_1,\rho_\eta]=\rho_{d_1 \eta}$ for any $\eta\in \Omega$, we find
    \begin{align*}
        [[d_1 K d_1, \rho_f], \rho_g]
        = & \, [[d_1,\rho_f]K d_1+d_1 K [d_1,\rho_f], \rho_g]\\
=&[\rho_{d_1 f}K d_1 +d_1K \rho_{d_1f}, \rho_g] \\
=&\rho_{d_1 f}K [d_1,\rho_g] +(-1)^{\deg_A(f)\deg_A(g)}[d_1,\rho_g]K \rho_{d_1f} \\
= &  \rho_{d_1 f} K \rho_{d_1 g} + (-1)^{\deg_A(f)\deg_A(g)} \rho_{d_1 g} K \rho_{d_1 f} \\
= &  \rho_{d f} K \rho_{d g} + (-1)^{\deg_A(f)\deg_A(g)} \rho_{d g} K \rho_{d f} \,
    \end{align*}
in $\mathrm{End}_k(M)$, where in the last step we used that $d_1$ is the de Rham differential $d$. Showing that $\pi\circ \varphi_{\{f,g\}}\circ j=0$ is therefore equivalent to showing that, for any $\xi\in V_0$, we have
\[
\pi(d f\wedge  K (d g\wedge \xi) + (-1)^{\deg_A(f)\deg_A(g)} d g\wedge K(d f\wedge \xi))=\{f,g\}_\Pi \xi.
\]
 By $A$-linearity of $\rho_{\{f,g\}}$ and by the local expression \eqref{eq:local}, we may assume $\xi\bigr\vert_U=dq^1\wedge \cdots \wedge dq^n$ in a Darboux chart $U$. Recalling that $d_1$ is the de Rham differential, and noticing that the de Rham differential maps $V_0$ to $V_1$, we find
\begin{align*}
\pi(d f\wedge  K (d g\wedge \xi))\biggr\vert_U&=\sum_{i=1}^n df\wedge \iota_\Pi (\frac{\partial g}{\partial p_i} dp_i\wedge dq^1\wedge \cdots \wedge dq^n)\\
&=\sum_{i=1}^n (-1)^{\deg_A(g)+i}\pi (df\wedge \frac{\partial g}{\partial p_i} \wedge dq^1\wedge \cdots \widehat{d q^i} \wedge\cdots \wedge dq^n)\\
&=\sum_{i=1}^n \frac{\partial f}{\partial q^i} \frac{\partial g}{\partial p_i} dq^1\wedge \cdots  \wedge dq^n.
\end{align*}
Hence,
\[
\pi(d f\wedge  K (d g\wedge \xi) + (-1)^{\deg_A(f)\deg_A(g)} d g\wedge K(d f\wedge \xi))\biggr\vert_U=
\sum_{i=1}^n \left(\frac{\partial f}{\partial q^i} \frac{\partial g}{\partial p_i}+ (-1)^{\deg_A(f)\deg_A(g)}\frac{\partial g}{\partial q^i} \frac{\partial f}{\partial p_i}\right)\xi\biggr\vert_U.
\]
Therefore we have $\{f,g\}\bigr\vert_U=\{f,g\}_\Pi\bigr\vert_U$ on every Darboux chart, and so $\{f,g\}=\{f,g\}_\Pi$ on $\mathbf{M}$.
\end{proof}
\begin{corollary}\label{cor:towards-bv-from-compatible-half-densities}
 A $d_2$-closed basis element $\zeta$ for $E_2$ defines a $BV_0$-operator $\Delta_\zeta$ enhancing the canonical $P_0$-algebra structure on  $C^\infty(\mathbf{M})$ to a $BV_0$-algebra structure, and so the canonical Gerstenhaber algebra structure on  $C^\infty(\mathbf{M})$ with the opposite grading to a Batalin-Vilkovisky algebra structure. 
\end{corollary}

{
Under the canonical identification $E_2\cong \Dens^{1/2}(\mathbf{M})$, a $d_2$-closed basis element $\zeta$ for $E_2$ is a nowhere vanishing half-density $\zeta$ such that $\Delta_{\mathrm{Dens}^{1/2}}\eta=0$. Half-densities satisfying this vanishing condition are called  ``compatible half-densities'' in the terminology of \cite{Cattaneo_2023} (see, in particular, \cite[Corollary 5.14(3)]{Cattaneo_2023}). By using this terminology, we obtain the following reformulation of Corollary \ref{cor:towards-bv-from-compatible-half-densities}.
\begin{corollary}\label{cor:bv-from-compatible-half-densities}
 A nowhere vanishing compatible half-density $\zeta$ on the odd symplectic manifold $(\mathbf{M},\omega)$  defines a $BV_0$-operator $\Delta_\zeta$ enhancing the canonical $P_0$-algebra structure on  $C^\infty(\mathbf{M})$ to a $BV_0$-algebra structure, and so the canonical Gerstenhaber algebra structure on  $C^\infty(\mathbf{M})$ with the opposite grading to a Batalin-Vilkovisky algebra structure. 
\end{corollary}

\begin{remark}
    A prototypical example of odd symplectic manifold is given by $\mathbf{M} =T^*[-1]\mathbf{N}$, with $\mathbf{N}$ an oriented smooth manifold. In this case, smooth functions on $\mathbf{M}$ are polyvector fields on $\mathbf{N}$ and $(E_2,d_2)=(\mathcal{E}^\bullet(\mathbf{N})[\dim \mathbf{N}],d)$. This way one recovers Example \ref{ex.real}. See \cite{Severa:2026juf} for details.
\end{remark}
}

\section{Outlook: derived Cartan differential operators over a dg-commutative algebra}

The recursive definition of differential operators from Section \ref{sec:1} can be expressed as a sequence of iterated pullbacks. Namely, let us denote by 
\[
\mathrm{ad}\colon \Hom_k(M,N)\to \Hom_k(A,\Hom_k(M,N))
\]
the linear map defined by setting
\[
\mathrm{ad}_\Bc\colon a\mapsto [\Bc,\rho_a].
\]
Then $\Diff^{\leq 0}_A(M,N)$ is defined as the pullback of $k$-vector spaces
\[
\begin{tikzcd}
    {\Diff^{\le 0}_A(M,N)} \ar[r] \arrow[d] & 0 \ar[d] \\
\Hom_k(M,N) \ar[r]& \Hom_k(A,\Hom_k(M,N))
\end{tikzcd}
\]
and then $\Diff^{\le n}_A(M,N)$ is recursively defined as the pullback
\[
\begin{tikzcd}
    {\Diff^{\le n}_A(M,N)} \ar[r] \arrow[d] & \Hom_k(A,\Diff^{\le n-1}_A(M,N)) \ar[d] \\
\Hom_k(M,N) \ar[r]& \Hom_k(A,\Hom_k(M,N)).
\end{tikzcd}
\]
By the universal property of pullbacks and by the Hom-tensor adjunction, one can paste these diagrams into a single big diagram where each square is a pullback:
\[
\begin{tikzcd}
    {\Diff^{\le n}_A(M,N)} \ar[r] \arrow[d] & \Hom_k(A,\Diff^{\le n-1}_A(M,N)) \ar[d]
    \ar[r] \arrow[d] & \Hom_k(A^{\otimes 2},\Diff^{\le n-2}_A(M,N)) \ar[d]\ar [r] &\cdots \\
\Hom_k(M,N) \ar[r]& \Hom_k(A,\Hom_k(M,N))\ar[r]& \Hom_k(A^{\otimes 2},\Hom_k(M,N))\ar[r]&\cdots
\end{tikzcd}
\]
and similarly for Cartan differential operators. While this formulation may appear to be a use of fancy categorical jargon for the sake of itself, it turns out to be quite useful when one wants to generalize the definitions of Grothendieck and Cartan differential operators to a derived setting, where $A$ is a commutative dg-algebra and $M$ and $N$ are dg-$A$-modules. Indeed, in order to get this generalization one replaces the Hom spaces $\Hom_k(-,-)$ with the Hom complexes $\mathbb{H}\mathrm{om}_k$, and considers homotopy pullbacks\footnote{The appropriate notion of pullback in a derived setting.} instead of strict pullbacks. For instance, doing so, and using the shifted dg-vector space $\mathfrak{g}_A=A[-1]$ in place of $A$, we see that the dg-vector space $\mathbb{C}\mathrm{artan}^{\leq 2}_A(M,M)$ of derived differential Cartan operators of order $\leq 2$ is defined by the iterated homotopy pullback
\[
\begin{tikzcd}
    {\mathbb{C}\mathrm{artan}^{\le 2}_A(M,M)} \ar[r] \arrow[d] & \mathbb{H}\mathrm{om}_k(\mathfrak{g}_A,\mathbb{C}\mathrm{artan}^{\le 1}_A(M,M)[-1]) \ar[d] \ar[ld,Rightarrow]
    \ar[r] \arrow[d] & \mathbb{H}\mathrm{om}_k(\mathfrak{g}_A^{\otimes 2},\mathfrak{g}_A[-1]) \ar[d] \ar[ld,Rightarrow]\\
\mathbb{H}\mathrm{om}_k(M,M) \ar[r]& \mathbb{H}\mathrm{om}_k(\mathfrak{g}_A,\mathbb{H}\mathrm{om}_k(M,M)[-1])\ar[r]& \mathbb{H}\mathrm{om}_k(\mathfrak{g}_A^{\otimes 2},\mathbb{H}\mathrm{om}_k(M,M)[-2])
\end{tikzcd}
\]
Taking cohomology of the outer diagram 
\[
\begin{tikzcd}
    {\mathbb{C}\mathrm{artan}^{\le 2}_A(M,M)} \ar[r] \arrow[d]  & \mathbb{H}\mathrm{om}_k(\mathfrak{g}_A^{\otimes 2},\mathfrak{g}_A[-1]) \ar[d] \ar[ld,Rightarrow]\\
\mathbb{H}\mathrm{om}_k(M,M) \ar[r] & \mathbb{H}\mathrm{om}_k(\mathfrak{g}_A^{\otimes 2},\mathbb{H}\mathrm{om}_k(M,M)[-2])
\end{tikzcd}
\]
and using the fact we are working over a field $k$, we get the commutative diagram of graded $k$-vector spaces
\[
\begin{tikzcd}
    H({\mathbb{C}\mathrm{artan}^{\le 2}_A(M,M)}) \ar[r] \arrow[d]  & \Hom_k(H(\mathfrak{g}_A)^{\otimes 2},H(\mathfrak{g}_A)[-1]) \ar[d] \\
\Hom_k(H(M),H(M)) \ar[r] & \Hom_k(H(\mathfrak{g}_A)^{\otimes 2},\Hom_k(H(M),H(M))[-2])
\end{tikzcd}
\]
and, so, by the universal property of the pullback, a canonical map
\[
H({\mathbb{C}\mathrm{artan}^{\le 2}_A(M,M)})\to \Cartan^{\leq 2}_{H(A)}(H(M),H(M)).
\]
In particular, degree 1 cocycles in ${\mathbb{C}\mathrm{artan}^{\le 2}_A(M,M)}$ induce degree 1 Cartan operators of degree $\leq 2$ on $H(M)$ as an $H(A)$-module. Expanding the definition of ${\mathbb{C}\mathrm{artan}^{\le 2}_A(M,M)}$ as an iterated homotopy pullback, we see that a degree 1 cocycles in the dg-module of derived differential Cartan operators of order $\leq 2$ on a dg-$A$-module $(M,\bc)$ is the datum of
a linear operator 
\[\Bc\colon M\to M[1]\]
such that
\begin{equation}\label{eq:B-NTT}
[\bc,\Bc]=0,
\end{equation}
together with a linear map
\[
\mathcal{L}\colon \mathfrak{g}_A\otimes M\to M
\]
such that for any $a\in A$ the morphism $\mathcal{L}_a$ is a derived differential Cartan operator of order $\leq 1$, which is a cocycle, and a homotopy $\mathcal{S}$ between $\mathrm{ad}_\Bc$ and $\mathcal{L}$, i.e., a linear operator 
\[ 
\mathcal{S}\colon \mathfrak{g}_A\otimes M\to M[-1]
\]
such that, for any $a\in A$, one has
$\mathcal{L}_a-[\Bc,\iota_a]=[\bc,\mathcal{S}_a]+\mathcal{S}_{d_{\mathfrak{g}_A}a},
$
i.e., equivalently,
\begin{equation}\label{eq:lie-NTT}
\mathcal{L}_a=[\Bc,\iota_a]+[\bc,\mathcal{S}_a]+\mathcal{S}_{d_{\mathfrak{g}_A}a}.
\end{equation}
The fact that $\mathcal{L}$ is a cocycle means that
\begin{equation}\label{eq:lie2-NTT}
[\bc,\mathcal{L}_a]=\mathcal{L}_{d_{\mathfrak{g}_A}}
\end{equation}
for any $a$ in $\mathfrak{g}_A$, while 
the fact that $\mathcal{L}$ takes valued in derived differential Cartan operators of order $\leq 1$ means there exist a degree zero\footnote{I.e., a degree 1 bilinear map $\mathfrak{g}_A^{\otimes 2}\to \mathfrak{g}_A[-1]$.}
\[
\{\,,\,\}\colon \mathfrak{g}_A\otimes \mathfrak{g}_A\to \mathfrak{g}_A
\]
that is a cocycle, i.e., such that
\begin{equation}\label{eq:bracket-NTT}
 d_{\mathfrak{g}_A}\{a,b\}=   \{d_{\mathfrak{g}_A}a,b\}+(-1)^{\deg_{\mathfrak{g}_A(a)}}\{a,d_{\mathfrak{g}_A}b\}
\end{equation}
for any $a,b\in \mathfrak{g}_A$,
and a homotopy $\mathcal{T}$ between $\mathrm{ad}_\mathcal{L}$ and $\iota_{\{\,,\,\}}$, that is a linear operator\footnote{An operator of this kind first appeared in \cite{gdt} and is therefore called a Gelfan’d-Daletski\u{\i}-Tsygan homotopy in \cite{fiorenza-kowalzig2018}.}
\[
\mathcal{T}\colon \mathfrak{g}_A\otimes \mathfrak{g}_A\otimes M\to M[-2]
\]
such that
\begin{equation}\label{eq:T-NTT}
 [\mathcal{L}_a,\iota_b]-\iota_{\{a,b\}}=[\bc,\mathcal{T}_{a,b}] -\mathcal{T}_{d_{\mathfrak{g}_A}a,b}-(-1)^{\mathrm{deg_{\mathfrak{g}_A(a)}}}\mathcal{T}_{a,d_{\mathfrak{g}_A}b}.  
\end{equation}
Equations \eqref{eq:B-NTT}-\eqref{eq:T-NTT} show that noncommutative differential calculi in the sense of Nest--Tamarkin--Tsygan \cite{NestTsyganCohomology,Nest-Tsygan,TamarkinTsygan} define order $\leq 2$ degree 1 cocycles in $\mathbb{C}\mathrm{artan}_A(M,M)$, with the additional property $\Bc^2=0$, and so they induce an element, that we will denote by the same symbol,
\[
\Bc\in \Cartan_{H(A)}^{\leq 2}(H(M),H(M))
\]
with $\Bc^2=0$. From this and Propositions \ref{prop:p0} and \ref{prop:bv0-module} we obtain the following.
\begin{proposition}\label{prop:menichi-type}
 Let $(A,M,\iota,\mathcal{S},\mathcal{T},\bc,\Bc)$ be a noncommutative differential calculus in the sense of Nest--Tamarkin--Tsygan. If $H(M)$ is a rank 1 free $H(A)$-module, then the graded commutative algebra structure of $H(A)$ gets enriched to a $P_0$-algebra structure. If $\zeta$ is a basis element for the   free rank 1 $H(A)$-module $H(M)$ such that $\Bc(\zeta)=0$, then this $P_0$-algebra structure is further enriched to a $BV_0$-algebra structure. By taking opposite degrees on $H(A)$ one gets a Gerstenhaber and a Batalin--Vilkovisky algebra structure, respectively.  
\end{proposition}
\begin{remark}
    As in Remark \ref{rem:palladio}, one can express the condition that $\zeta$ is a basis element for the  rank 1 free $H(A)$-module $H(M)$ such that $\Bc(\zeta)=0$ in terms of a representative $\xi$ of $\zeta$ in $M$ by requiring $\xi$ is a Palladio element. 
\end{remark}
\begin{remark}
A natural generalization of the above construction consists in taking a more general $\mathrm{Comm}_\infty$-algebra $A$ in place of a dgca $A$ and a  $\mathrm{Comm}_\infty$-module $M$ for the  $\mathrm{Comm}_\infty$-algebra $A$, i.e., a dg-module $(M,\bc)$ equipped with an $A_\infty$-morphism $\rho\colon A\to (\mathrm{End}_k(M),[\bc,-])$, instead of a dg$-A$-module $M$. Doing so, the theory of derived Cartan operators is general enough to contain the noncommutative differential calculus on the Hochschild complex of an associative algebra and more generally of a cyclic operad with multiplication. The $\mathrm{Comm}_\infty$-version of Proposition \ref{prop:menichi-type} then reproduces Menichi's Batalin-Vilkovisky algebra structure on Hochschild cohomology of Calabi-Yau algebras \cite{Menichi2007} as well as the generalization of this result to the cohomology of a cyclic operad with multiplication. See \cite{fiorenza-kowalzig2020} for an overview of these results and additional details. The rigorous construction of derived Cartan operators for $\mathrm{Comm}_\infty$-algebras goes beyond the aim of the present article, so we content ourselves with this sketchy remark here.
\end{remark}

\appendix
\section{Appendix: Xu modules}

Grothendieck and Cartan differential operators naturally sheafify: if $\mathbf{X}$ is a topological space,  $\A$ a sheaf of graded-commutative
algebras on $\mathbf{X}$, and let $\M$ be a sheaf of graded $\A$-modules one has the sheaf 
$\Diffscr_{\A}^{\leq k}(\M,\M)$ and, when $\M$ is faithful, its subsheaf $\Cartanscr_{\A}^{\leq k}(\M,\M)$.
Passing to global sections, and setting $
A=\Gamma(\mathbf{X},\A)$, and $
M=\Gamma(\mathbf{X},\M)$,
one sees that a
global section $
\Bc\in\Gamma\bigl(X,\Diffscr_{\A}^{\leq k}(\M,\M)\bigr)$
acts as a differential operator of order at most $k$ on the $A$-module $M$. This gives a canonical map 
$
\Gamma\bigl(\mathbf{X},\Diffscr_{\A}^{\leq k}(\M,\M)\bigr)
\longrightarrow
\mathrm{Diff}_A^{\leq k}(M,M)
$
inducing a canonical map
$
\Gamma\bigl(\mathbf{X},\Cartanscr_{\A}^{\leq k}(\M,\M)\bigr)
\longrightarrow
\mathrm{Cartan}_A^{\leq k}(M,M).
$
One calls the global sections of $\Diffscr_{\A}^{\leq k}(\M,\M)$ and of $\Cartanscr_{\A}^{\leq k}(\M,\M)$ differential operators and Cartan differential operators of order at most $k$ on $\M$, respectively.
One generalizes Definition \ref{def.P0-module} by saying that a $P_0$-module for $\A$ is a pair $(\M,\Bc)$ where $\M$ is a locally free $\A$-module of rank 1 and $\Bc$ is a differential operator of order at most 2 and degree $1$ on $\M$ such that $\Bc^2=0$. A $P_0$-module $(\M,\Bc)$ for $\A$ induces on $\A$ the structure of a sheaf of $P_0$ algebras. 

In order to reproduce Xu's construction from \cite{Ping} within the framework of $P_0$-modules, we need a suitable notion of connection. This is done as follows. 
Let us write $\mathfrak{g}_\A$ for the associated sheaf of graded Lie algebras, $\mathfrak{t}_\A=\mathfrak{g}_\A^0$ for the sheaf of Lie algebras given by the degree zero component of $\mathfrak{g}_\A$.
The sheaf of graded commutative algebras $\A$ is naturally a sheaf of $\mathfrak{g}_\A$-modules, and the subsheaf of commutative algebras $\A^0$ given by the degree zero component of $\A$ is a sheaf of $\mathfrak{t}_\A$-modules. We write
\[
(\Omega^\bullet_\A,d_{\Omega_\A})=(\mathrm{CE}(\mathfrak{t}_\A;\A^0), d_\mathrm{CE})
\]
for the sheaf of differential graded commutative algebras given by the Chevalley-Eilenberg complexes of $\mathfrak{t}_\A$ with coefficients in the $\mathfrak{t}_\A$-module $\A^0$.
We write $\Omega^\bullet_\A(\M^0)$ for the sheaf of graded $\Omega^\bullet_\A$-modules defined by
\[
\Omega^\bullet_\A(\M^0)=\Omega^\bullet_\A\otimes_{\A^0}\M^0.
\]
\begin{definition}
A $\mathfrak{t}_\A$-connection on $\M^0$ is a degree 1 linear map
\[
\nabla\colon \Omega^\bullet_\A(\M^0)\to \Omega^{\bullet+1}_\A(\M^0)
\]
such that 
\[
\nabla(\alpha m)=d_{\Omega_\A}(\alpha) m+(-1)^{
\deg(\alpha)}\alpha \nabla(m).
\]
We say that $\nabla$ is flat if $\nabla^2=0$.
\end{definition}
\begin{definition}
A \emph{Xu structure} on a $P_0$-module $(\M,\Bc)$ is a $\Omega^\bullet_\A$-dg-module structure 
\begin{align*}
\Omega^\bullet_\A\otimes \M&\to \M\\
\alpha\otimes m&\mapsto \alpha\bullet m
\end{align*}
extending the $\A^0$-module structure of $\M$ and such that for any $\alpha \in \Omega^k_{\A}$ the operator $\mu_\alpha=\alpha\bullet -$ is a differential operator of order at most $k$ with respect to the $\A$-module structure of $\M$.
\end{definition}
\begin{remark}
Since $\Omega^\bullet_{\A}$ is generated by $\Omega^1_{\A}$ over $\A^0$, and since $\Diffscr_{\A}(\M,\M)$ is a filtered algebra, it suffices to ask that for any $\alpha \in \Omega^1_{\A}$ the operator $\mu_\alpha=\alpha\bullet -$ is a differential operator of order at most $1$. 
\end{remark}
\begin{example}
Let $\mathbf{E}\to \mathbf{M}$ be a Lie algebroid, let $\mathcal{E}$ be its sheaf of sections, and let $\mathcal{A}$ be the sheaf of sections of $\wedge^\bullet \mathbf{E}$ and $\M$ be the sheaf of sections of $\wedge^\bullet \mathbf{E}^*$, with degrees shifted by the rank of $\mathbf{E}$. Then we have a natural identification of $\M$ with $\Omega_\A[\mathrm{rk}\, \mathbf{E}]$ in this case, and $\Bc=d_{\Omega_\A}$ makes $(\M,\Bc)$ a $P_0$-module for $\A$. The $\Omega^\bullet_{\A}$-module structure on $\M$ is the one coming from the identification between $\M$ and $\Omega^\bullet_{\A}$. To see that this is indeed a Xu structure, let $\alpha\in \Omega^1_\A$. In the situation considered in this example the action of $\A$ on $\M$ is given by contractions, and contractions by local sections of $\mathbf{E}$ act as derivations. From this and the fact that $\alpha\in \Omega^1_\A$ we see that for any local section $Y$ in $\A$ the commutator $[\iota_Y,\alpha\wedge]$ is of the form $
[\iota_Y,\alpha\wedge]=\iota_Z$
for some local section $Z\in \A$. Since $[\iota_Z,\iota_W]=0$ for any $W$ in $\A$, this gives $[\iota_W,[\iota_Y,\mu_\alpha]]=0$, i.e., $\mu_\alpha$ is a first order operator. Finally, since $\Bc=d_{\Omega_\A}$, the dg-module condition $\Bc(\alpha\bullet m)=(d_{\Omega_\A}\alpha)\bullet m+(-1)^{\deg \alpha}\alpha\bullet \Bc(m)$ is trivially satisfied.
\end{example}
The last definition we need is the following.
\begin{definition}
Let $\A$ be a sheaf of graded-commutative algebras on the topological space $\mathbf{X}$. A \emph{Xu module} for $\A$ is a quintuple $(\M,\Bc,\{\zeta_U\},\nabla,\bullet)$ where: $(\mathcal{M},\Bc)$ is a $P_0$-module for $\A$; $\{\zeta_U\}$ is a collection of local sections of $\M$ such that $(\mathcal{M}\vert_U,\Bc\vert_U,\zeta_U)$ is a $BV_0$-module for $\A\vert_U$; $\nabla$ is a flat $\mathfrak{t}_\A$-connection on $\M^0$; $\bullet$ is a Xu structure on $\M$.  
\end{definition}
\begin{proposition}
The datum of a Xu module $(\M,\Bc,\{\zeta_U\},\nabla,\bullet)$ for $\A$ induces a globally defined  $BV_0$-operator $\Delta^{\mathrm{Xu}}\in \Gamma\bigl(\mathbf{X},\Diffscr_{\A}^{\leq 2}(\A,\A)\bigr)$ compatible with the $P_0$-algebra structure on $\A$. If moreover $\M$ is concentrated in nonpositive degrees, then the \emph{Xu Laplacian} $\Delta^{\mathrm{Xu}}$ satisfies $\Delta^{\mathrm{Xu}}(1)=0$ and so enhances the $P_0$-algebra structure on $\A$ to a $BV_0$-algebra structure.  
\end{proposition}
\begin{proof}
Let $\Xi_{\zeta_U^{}}\in \Omega^1_{\A}\vert_U$ be the connection 1-form defined by the equation
\[
\nabla \zeta_U=\Xi_{\zeta_U^{}}\zeta_U.
\]
Since the connection $\nabla$ is flat, we have $d_{\Omega_\A}\Xi_{\zeta_U^{}}=0$. Indeed,
 \[
 0=\nabla^2(\zeta_U)=\nabla(\Xi_{\zeta_U^{}}\zeta_U)=(d_{\Omega_\A}\Xi_{\zeta_U^{}})\zeta_U-(\Xi_{\zeta_U^{}}\wedge \Xi_{\zeta_U^{}})\zeta_U=(d_{\Omega_\A}\Xi_{\zeta_U^{}})\zeta_U.
 \]
 This, using again the fact that $\Xi_{\zeta_U^{}}\wedge \Xi_{\zeta_U^{}}=0$ since $\Xi_{\zeta_U^{}}$ has total degree 1, tells us that the operator 
 \[
 \Bc-\mu_{\Xi_{\zeta_U^{}}}\colon \M\vert_U\to \M\vert_U
 \]
 is a $BV_0$-operator. Indeed, we have
 \[
 (\Bc-\mu_{\Xi_{\zeta_U^{}}})^2=\Bc^2-[\Bc,\mu_{\Xi_{\zeta_U^{}}}]+\mu_{\Xi_{\zeta_U^{}}\wedge \Xi_{\zeta_U^{}}}=-\mu_{d_{\Omega_\A}\Xi_{\zeta_U^{}}}=0,
 \]
where we used that $\bullet$ is a $\Omega_\A^\bullet$-dg-module structure on $(\M,\Bc)$,
while the fact that $(\Bc-\mu_{\Xi_{\zeta_U^{}}})$ has order at most 2 is an immediate consequence of the fact that $\bullet$ is a Xu structure. 
 Since $\Bc-\mu_{\Xi_{\zeta_U^{}}}$ differs from $\Bc$ by an operator of order at most 1, the operators $\Bc-\mu_{\Xi_{\zeta_U^{}}}$  and $\Bc$ induce the same $P_0$-bracket $\{\,,\,\}$ on $\A\vert_U$.

 Let us write
\[
\varphi_{\zeta_U}:\M|_U\longrightarrow\A|_U
\]
for the distinguished isomorphism of $\A|_U$-modules determined by the condition
\[
\varphi_{\zeta_U}(\zeta_U)=1.
\]
We can use the isomorphisms $\varphi_{\zeta_U}$ to translate the $BV_0$-operator $\Bc-\mu_{\Xi_{\zeta_U^{}}}$ to a $BV_0$-operator 
\[
\Delta^{\mathrm{Xu}}_U=\varphi_{\zeta_U}\circ (\Bc-\mu_{\Xi_{\zeta_U^{}}})\circ \varphi_{\zeta_U}^{-1}=
\Delta_{\Bc,\zeta_U}-\xi_{\zeta_U^{}}\colon \A\vert_U\to \A\vert_U,
\]
where we have written 
\[
\Delta_{\Bc,\zeta_U}=\varphi_{\zeta_U}\circ \Bc\circ \varphi_{\zeta_U}^{-1}; \qquad \xi_{\zeta_U^{}}=\varphi_{\zeta_U}\circ \mu_{\Xi_{\zeta_U^{}}}\circ \varphi_{\zeta_U}^{-1}.
\]
As the notation suggests, the two summands $\Delta_{\Bc,\zeta_U}$ and $\xi_{\zeta_U^{}}$ depend on the choice of $\zeta_U$, but $\Delta^{\mathrm{Xu}}_U$ is independent of this choice. Therefore the locally defined Xu Laplacians $\Delta^{\mathrm{Xu}}_U$ define a global section of  $\Diffscr_{\A}^{\leq 2}(\A,\A)$. Indeed, by  independence of the choice of local generator, both $\Delta^{\mathrm{Xu}}_U$ and $\Delta^{\mathrm{Xu}}_V$ restrict to $\Delta^{\mathrm{Xu}}_{U\cap V}$ on $U\cap V$.
To see that $\Delta^{\mathrm{Xu}}_U$ is independent of $\zeta_U$, recall from Remark \ref{laplacian:basis-choice}
 that if $\zeta_U'$ is another choice for a local basis element, then 
\[
\Delta_{\Bc,\zeta_U'}=\Delta_{\Bc,\zeta_U}-x_U^{-1}\{x_U,-\},
\]
where $x_U\in \A^0\vert_U$ is the degree zero invertible element such that $\zeta'_U=x_U\zeta_U$. We are then reduced to showing that
\[
\xi_{\zeta_U'}=\xi_{\zeta_U^{}}-x_U^{-1}\{x_U,-\},
\]
The defining equations for $\xi_{\zeta_U^{}}$ and $\xi_{\zeta_U'}$ are
\[
\xi_{\zeta_U^{}}(a_U) \zeta_U=\Xi_{\zeta_U^{}}\bullet (a_U\zeta_U)\qquad\text{and}\qquad \xi_{\zeta'_U}(a_U) \zeta'_U=\Xi_{\zeta'_U}\bullet (a_U\zeta'_U),
\]
respectively. Since $\bullet$ extends the $\A^0$-module structure on $\M$, we have
\[
x_U^{-1}\Xi_{\zeta'_U}\bullet (a_U\zeta'_U)=\Xi_{\zeta'_U}\bullet (a_Ux_U^{-1}\zeta'_U)=\Xi_{\zeta'_U}\bullet (a_U\zeta_U),
\]
and so we are reduced to showing that
\[
\Xi_{\zeta'_U}\bullet (a_U\zeta_U)=\Xi_{\zeta_U^{}}\bullet (a_U\zeta_U)-x_U^{-1}\{x_U,a_U\}\zeta_U.
\]
We have 
\[
x_U\Xi_{\zeta'_U}\zeta_U= \Xi_{\zeta'_U}\zeta'_U=\nabla\zeta'_U=\nabla(x_U\zeta_U)=d_{\Omega_\A}(x_U)\zeta_U+x_U\nabla\zeta_U=d_{\Omega_\A}(x_U)\zeta_U+x_U\Xi_{\zeta_U^{}}\zeta_U,
\]
hence
\[
\Xi_{\zeta'_U}=\Xi_{\zeta_U^{}}+x^{-1}_Ud_{\Omega_\A}(x_U).
\]
Therefore, we are reduced to proving that
\[
d_{\Omega_\A}(x_U)\bullet (a_U\zeta_U)=-\{x_U,a_U\}\zeta_U.
\]
Since $\Bc \zeta_U=\Bc \zeta'_U=0$, we have $[\Bc,\iota_{a_U}]\zeta_U=\Bc(a_U\zeta_U)$ and $[\Bc,\iota_{a_U}](x_U\zeta_U)=\Bc(a_Ux_U\zeta_U)=\Bc(x_Ua_U\zeta_U)$, and so
\begin{align*}
\{x_U,a_U\}\zeta_U&=[[\iota_{x_U},\Bc],\iota_{a_U}]\zeta_U\\&
=[\iota_{x_U},[\Bc,\iota_{a_U}]]\zeta_U\\&
=x_U\Bc(a_U\zeta_U)-\Bc(x_Ua_U\zeta_U)\\
&=-d_{\Omega_\A}(x_U)\bullet(a_U\zeta_U),
\end{align*}
where in the last step we used that $\bullet$ is a dg-module structure. 
Finally, if $\M$ is concentrated in nonpositive degrees, then by evident degree reasons we have $
 \mu_{\Xi_{\zeta_U^{}}}(\zeta_U)=0$. This gives $\Delta^{\mathrm{Xu}}_U(1)=0$ for any $U$, and so $\Delta^{\mathrm{Xu}}(1)=0$.
\end{proof}

\section*{Acknowledgments}
The authors thank
the Mittag-Leffler Institute for its hospitality during the 2025 program {\it Cohomological Aspects of Quantum
Field Theory}. 
We wish to express our gratitude to P.~\v{S}evera, M.~Cueca Ten, J.~Pulmann, F.~Bonechi, A.~Cattaneo, E.~Getzler for stimulating interactions. We would like to thank the organisers of the "Fifty Years of BRST" conference in LMU Munich, especially I.~Sachs, for bringing us together in March 2026. E.B.~is supported by GA\v{C}R grant PIF-OUT 24-10634O. D.F.'s research is part of the activities of the Department of Excellence (Dipartimento di Eccellenza) Project CUP B83C23001390001. D.F.~is a member of the Gruppo Nazionale per le Strutture Algebriche,
Geometriche e le loro Applicazioni (GNSAGA–INdAM).

\bibliographystyle{alpha}
\bibliography{biblio}

\end{document}